\documentclass[aps,pra,amsmath,amssymb,reprint,10pt,longbibliography,onecolumn,notitlepage]{revtex4-1}
\usepackage{graphicx,mathrsfs,enumerate}
\usepackage[utf8]{inputenc}
\usepackage{mathtools}
\usepackage{amsfonts,amstext}
\usepackage{amsmath,amsthm,amssymb,mathbbol,tensor}
\usepackage{graphicx,xcolor,mathrsfs,enumerate}
\usepackage[hidelinks]{hyperref}\hypersetup{pdfpagemode=UseNone,pdfstartview=}
\usepackage[left=2cm,right=2cm]{geometry}
\usepackage{url}
\usepackage{times}
\usepackage{bbm}
\usepackage[english]{babel}

\newcommand{\bra}[1]{\langle #1|}
\newcommand{\ket}[1]{|#1\rangle}
\newcommand{\braket}[2]{\langle #1|#2\rangle}
\newcommand{\ketbra}[2]{| #1 \rangle \langle #2 |}

\newtheorem{thm}{Theorem}
\newtheorem{cor}[thm]{Corollary}
\newtheorem{lem}[thm]{Lemma}

\begin{document}
\title{Error-run-time trade-off in the adiabatic approximation beyond scaling relations}
\author{M. R. Passos}\affiliation{Instituto de F\'isica, Universidade Federal do Rio de Janeiro, P. O. Box 68528, Rio de Janeiro, RJ 21941-972, Brazil}
\author{M. M. Taddei\footnote{marciotaddei[at]gmail.com}}\affiliation{Instituto de F\'isica, Universidade Federal do Rio de Janeiro, P. O. Box 68528, Rio de Janeiro, RJ 21941-972, Brazil}
\author{R. L. de Matos Filho} \affiliation{Instituto de F\'isica, Universidade Federal do Rio de Janeiro, P. O. Box 68528, Rio de Janeiro, RJ 21941-972, Brazil}

\begin{abstract}
The use of the adiabatic approximation in practical applications, as in adiabatic quantum computation, demands an assessment of the errors made in finite-time evolutions. 
Aiming at such scenarios, we derive bounds  relating  error and evolution time in the adiabatic approximation that go beyond typical scaling relations.
Using the Adiabatic Perturbation Theory, we obtain leading-order expressions valid for long evolution time $T$, while  explicitly determining the shortest time $T$ and the largest error $\varepsilon$ for which  they are valid. In this validity regime,  we can make clear and precise statements about the evolution time needed to reach a given error and vice-versa. As an example of practical importance, we apply these results to the adiabatic search, and obtain for the first time an error-run-time trade-off relation that fully reproduces the discrete-Grover-search scaling. We also pioneer the obtention of tight numerical values for $\varepsilon$ and $T$ under the error-reducing strategy  ``boundary cancelation''.
\end{abstract}

\maketitle

\section{Introduction}\label{sec:intro}

The quantum adiabatic theorem \cite{Born1927,Born1928} is a fundamental result derived early in the development of quantum mechanics. It states that a system   submitted to the action of a time-dependent Hamiltonian which changes sufficiently slowly in time will remain in an instantaneous eigenstate of that Hamiltonian throughout its evolution. Since its appearance, it has played an important role in quantum mechanics and  has found application in a wide range of subjects~\cite{Smith1969,Oreg1984,Averin1998,Shapiro2006}. More recently, the adiabatic theorem has been used as the basis for a new method of quantum computation, known as adiabatic quantum computing (AQC)~\cite{Farhi2000,Farhi2001,Albash2016}, generating a renewed interest in the quantum adiabatic approximation. 
In AQC, one chooses an initial Hamiltonian whose ground state can be easily prepared, whereas the corresponding eigenstate  of the final Hamiltonian encodes some information of interest (such as the answer to a computational problem). The time-dependent Hamiltonian is free to interpolate between the two so as to adiabatically drive the simple initial eigenstate into the information-encoding  eigenstate at the final time $T$ as well as possible. 
AQC is as powerful as the circuit version of quantum computation \cite{Aharonov2008,Albash2016} and is featured prominently in proposals for quantum machine learning \cite{Wittek2014}, quantum chemistry simulations \cite{Babbush2015}, as well as in several other examples of adiabatic quantum algorithms~\cite{Albash2016}. Quantum annealing~\cite{Finnila1993,Kadowaki1998} also bears a close relationship to AQC~\cite{McGeoch2014}.

As seen in its original form \cite{Born1928} or in standard textbook approaches \cite{Messiah1962}, however, the quantum adiabatic theorem only applies to infinitely slow system driving. Whereas in early applications it was not decisive  to have a rigorous bound on the shortest driving time $T$ necessary for  the adiabatic approximation  be reliable,  in AQC it is of fundamental importance. Since the main goal of  adiabatic quantum computing is to perform computations in a  time scale shorter than its classical counterparts, it becomes peremptory to have accurate estimates of the time necessary for an adiabatic behavior. 
This has taken the form of trade-off relations between the total time $T$ allowed for the system to be driven, which is called the run time of the algorithm, and the error made by applying the adiabatic theorem as an approximation for determining the final state of the system at the time $T$~\cite{Avron1987,Joye1991,Joye2002,Jansen2007,Lidar2009}.
The focus is to find an interpolation between initial and final Hamiltonians which minimizes  the run time $T$ and the error $\varepsilon$, measured by the distance between the actual state of the system at time $T$ and the ground state of the Hamiltonian at that time.   One, typically, seeks to minimize the run time $T$ for a fixed error $\varepsilon$ and problem size or, alternatively, to minimize the error $\varepsilon$ for a given run time $T$ and problem size.  

Since the appearance of AQC, bounds to the relation between  the run time $T$  and the error $\varepsilon$ have been presented~\cite{Albash2016}. On one hand, there are bounds of very general validity \cite{Ambainis2004,Jansen2007,OHara2008,Cheung2011}. These bounds, although rigorous, are not tight, and for this reason overestimate the run time for given error. Importantly, this may lead to incorrect scaling of the trade-off relation with system size, as we will later show. On the other hand, some authors have obtained relations between $\varepsilon$ and $T$ that, because they rely on some sort of approximation, are valid only in specific regimes \cite{MacKenzie2006,Rezakhani2010a,Wiebe2012}. Even though some of those are asymptotically tight, the lack of clear quantitative statements about the domain of validity of those approximations prevents one from establishing the range  of values of $T$ and $\varepsilon$ to which those relations in fact apply, leading to an inability to estimate the complexity of  quantum adiabatic algorithms (and in some cases, to wrong estimations). One of our main goals here is to provide tight  bounds to the relation between the run time $T$ and the error $\varepsilon$, whereas explicitly determining a lower limit on the value of $T$ and/or an upper limit on the value of $\varepsilon$ for which these bounds apply.

In this article we use the adiabatic perturbation theory (APT)~\cite{Rigolin2008} to obtain  general results for the  trade-off  relation between the run time $T$ and the error $\varepsilon$ of the adiabatic approximation. Going beyond scaling relations, we obtain leading-order expressions valid for large  time $T$. More importantly, we are able to explicitly determine the shortest time $T$ and the largest error $\varepsilon$ for which  these results are valid.  Restricting our considerations to within this validity regime,  we  can make clear and correct statements about the run time necessary to reach a given error  and vice-versa.
We apply our results to the quantum adiabatic search algorithm~\cite{Farhi2000,Roland2002} and obtain for the first time precise values for the trade-off relation, whose scaling with run time, error and system size in many instances has not been previously reported in the literature. In particular, we find for the first time that, for the adiabatic search problem under optimal driving, the error scales as $\varepsilon\sim1/\sqrt N$ with run time $T\sim\sqrt N$, exactly reproducing the scaling of the original, circuit-based Grover algorithm~\cite{Roland2003}. Furthermore, we obtain previously unknown error-run-time trade-off relations for the strategy called ``boundary cancelation'' \cite{Garrido1962,Lidar2009,Wiebe2012,Albash2016}, which imposes constraints on the driving to reduce errors.

This article is organized as follows. In Sec.~\ref{sec:main} we outline the framework of the adiabatic approximation and  derive general results for the trade-off relation between run time $T$ and error $\varepsilon$. In Sec.~\ref{sec:grover}, these results are applied to the adiabatic quantum search algorithm, which, besides being of  practical importance for AQC,  admits an exact analytical treatment and illustrates the capabilities of our method to derive novel results;  Sect \ref{sec:conclusions} is left to final remarks.

\section{Adiabatic errors and validity conditions}\label{sec:main}

We treat the evolution of a closed system from an initial time $t=0$ to a final time $t=T$ governed by a driven Hamiltonian of the form $H(t/T)$. This dependence solely on $s:=t/T$ ($\in [0,1]$) includes Hamiltonians of interest for most applications, most notably adiabatic quantum computation (AQC), but excludes those with more than one independent timescale. Importantly, this excludes Hamiltonians which can lead to so-called resonances \cite{Marzlin2004,Tong2005,Duki2006,Ma2006,Marzlin2006}, which happen when two different timescales of the Hamiltonian coincide \cite{Du2008,Ambainis2004,Tong2007,Comparat2009} and results in far-from-adiabatic evolution.

Besides being the final time, $T$ serves as a global parameter which rescales the total driving time, hence allows an adjustment of the speed of change of the Hamiltonian. The evolution obeys the Schrödinger equation, which, in terms of dimensionless time $s$, reads
\begin{equation}
\frac{i\hbar}{T} \ket{\dot\Psi(s,T)}= H(s) \ket{\Psi(s,T)} \ ,
\label{eq:Schroedinger}
\end{equation}
where $\ket{\Psi(s,T)}$ is the physical state of the system at $s$ and a dot over a variable indicates partial derivative with respect to $s$. The breakdown of the time dependence of each variable in $s,T$ is convenient to our discussion of long times $T$.

At each time $s$, there is an instantaneous eigenbasis $\{\ket{\phi_n(s)}\}$ of $H(s)$ with eigenenergies $E_n(s)$ nondecreasing with $n$  ($E_0(s)$ being the ground-state energy). The time dependence of each $\ket{\phi_n(s)}$ reflects the driving program and does \emph{not} obey the Schrödinger equation. Assuming an evolution that starts on the $j$-th eigenstate $\ket{\Psi(0,T)}=\ket{\phi_j(0)}$, however, the adiabatic theorem in its traditional form \cite{Messiah1962} states that, for infinitely long driving time $T$, the physical state coincides with the corresponding eigenstate along the evolution, $\ket{\Psi(s,T)}=e^{i\alpha}\ket{\phi_j(s)}$ for some phase $e^{i\alpha}$, given that $E_j(s)$ is non-degenerate throughout. The results of this paper will be valid for any initial eigenstate, but to fix the notation, we will consider, in the remainder of the text, the system to be initially in the ground state ($j=0$) --- and $E_0(s)$ will be assumed to be non-degenerate. 

The error of the adiabatic theorem amounts to precising how distant the physical state $\ket{\Psi(s,T)}$ is from the instantaneous ground state $\ket{\phi_0(s)}$.  Illustrative of their different roles is the fact that $\ket{\phi_0(s)}$ does not depend on $T$. Several figures of merit for this error $\varepsilon$ have been used in the literature \cite{Jansen2007,Lidar2009,Wiebe2012}. We will use an explicit well-known distance, the Bures angle:
\begin{equation}
\varepsilon=D[\ket{\Psi(s,T)},\ket{\phi_0(s)}] := \arccos \left(\left| \braket{\phi_0(s)}{\Psi(s,T)} \right|\right) \ .
\label{eq:def_Bures}
\end{equation}
The Bures angle has the advantage of being a genuine, Riemannian distance between the two states, defined from the Fubini-Study metric \cite{Bures1969,Kobayashi1969,Braunstein1994,Uhlmann1996,Taddei2014}, and, importantly, it depends on the fidelity between the two states, having a maximum value of $\pi/2$ for orthogonal states and a minimum of zero if, and only if, the two states are the same. This is particularly relevant in that we want not only to describe the scaling of the error, but to obtain numerical values.

In order to estimate the error made by using the adiabatic approximation, we make use of the Adiabatic Perturbation Theory \cite{Rigolin2008}, which relies on an expansion well-suited for long evolution times. First, the physical state of the system is decomposed in terms of the energy eigenbasis: 
\begin{equation}
\ket{\Psi(s,T)}=\sum_{n=0}e^{-iT\omega_n(s)}b_n(s,T)\ket{\phi_n(s)} \ ,
\label{eq:Psi_expansion}
\end{equation}
where the dynamical phase $\omega_n(s):=\frac1\hbar\int_0^sE_n(s')ds'$ and the sum extends over all eigenstates. The geometric phase usually appearing in this expansion is absent here because, without loss of generality, it will be taken as zero by assuming a choice of eigenvector phases such that $\braket{\dot\phi_n(s)}{\phi_n(s)}\equiv0$. 

To treat the long-run-time regime, each complex coefficient $b_n(s,T)$ is written \cite{Rigolin2008} as a sum of powers of $\frac1T$. Due to $T$-dependency of Eq.\eqref{eq:Schroedinger}, it is convenient to write it in terms of powers of $i\hbar/T$:
\begin{equation}
b_n(s,T) = \sum_{p=0}^\infty \frac{(i\hbar)^p}{T^p} b_n^{(p)}(s) \ .
\label{eq:bn_expansion}
\end{equation}
 In a slight abuse of notation, coefficients $b_n^{(p)}(s)$ may still depend on $T$ (usually in the form of a complex phase), but are upper- and lower-bounded by $T$-independent expressions. 
The (zero-order) adiabatic approximation amounts to taking $b_n(s)=b_n^{(0)}(s)=\delta_{n0}$, formally equivalent to the $T\to\infty$ limit. We will be interested in leading-order expressions beyond $b_n^{(0)}(s)$.

Importantly, any analysis based on a perturbative expansion requires validity conditions --- in our case, conditions for it to be truncated at the first nonvanishing term. To this end, we will obtain closed-form expressions for leading and next-to-leading order coefficients. This is essential for a correct analysis of the perturbative results.

\subsection{General results}\label{sec:general}

Let us begin to state our main results. The Bures angle can be expanded as
\begin{equation}
D[\ket{\Psi(s,T)},\ket{\phi_0(s)}] = \frac{\hbar}{T} \sqrt{\sum_{n\neq0}\left|b_n^{(1)}\right|^2}
- \frac{\hbar^2}{T^2} \left(\frac{\sum_{n\neq0}{\rm Im}\left(b_n^{*(1)}b_n^{(2)}\right)}{\sqrt{\sum_{n\neq0}\left|b_n^{(1)}\right|^2}}-\sqrt{\sum_{n\neq0}\left|b_n^{(1)}\right|^2} \ {\rm Im} \left(b_0^{(1)}\right) \right) + O\left(\frac1{T^3}\right) \ ,
\label{eq:DpowersTupto2}
\end{equation}
where, as before, the terms may have additional $T$ dependences, as long as they are both upper- and lower-bounded by $T$-independent expressions; $\sum_{n\neq0}$ indicates a sum that runs over all excited levels (proof in Appendix \ref{sec:expansionD1T}).
By calculating the first-order terms $b_n^{(1)}$, we can write a closed-form expression for the error:
\begin{equation}
\varepsilon=D[\ket{\Psi(s,T)},\ket{\phi_0(s)}] =  
\frac{\hbar}{T}\sqrt{\sum_{n\neq0} \left|\left[ \left.e^{iT\omega_{n0}}\frac{\braket{\phi_n}{\dot H|\phi_0}}{\Delta_{n0}^2}\right|_0^s\ \right]\right|^2}
 + O\!\left(\frac1{T^2}\right) \ , 
\label{eq:D_along}\end{equation}
where \mbox{$\Delta_{n0}(s):=E_n(s)-E_0(s)$}, $\omega_{n0}(s):=\omega_n(s)-\omega_0(s)$ and the notation $f|_a^b:=f(s=b)-f(s=a)$ has been used --- in this case, $f$ involves the gaps, inner products, phases.

For numerical applications, it is interesting to avoid the instability of calculating the oscillating exponential. This can be done by upper- and lower-bounding the leading-order term above: 
\begin{equation}
\varepsilon=D[\ket{\Psi(s,T)},\ket{\phi_0(s)}] \lessgtr 
\frac{\hbar}{T}\sqrt{\sum_{n\neq0} \left(\left| \frac{\braket{\phi_n}{\dot H|\phi_0}}{\Delta_{n0}^2}\right|_s \pm \left|\frac{\braket{\phi_n}{\dot H|\phi_0}}{\Delta_{n0}^2}\right|_{0}\right)^2} + O\!\left(\frac1{T^2}\right) \ , 
\label{eq:D_along_bounds}\end{equation}
where the $+$/$-$ sign gives the upper/lower bound. These bounds also have the advantage of only depending on initial- and final-time parameters, since they do away with $\omega_{n0}(s)$.

The $1/T^2$ term in Eq.\eqref{eq:DpowersTupto2} can be used to establish a reasonable, practical validity condition for Eqs.(\ref{eq:D_along},\ref{eq:D_along_bounds}): they will be valid as long as $T$ is large enough that the second term is negligible compared to the first. This happens when
\begin{equation}
T\geqslant T_{\rm val} = C \ \hbar \left|\frac{\sum_{n\neq0}{\rm Im}\left(b_n^{*(1)}(s)\ b_n^{(2)}(s)\right)}{\sum_{n\neq0}\left|b_n^{(1)}(s)\right|^2}\right| \ ,
\label{eq:validity_bn}
\end{equation} 
where  $C$ is a large constant (whose actual value is accuracy dependent) and the fact that $b_0^{(1)}\in\mathbb R$ has been used. 
Importantly, from $T_{\rm val}$ one can obtain an upper bound $\tilde{\varepsilon}$ on the error $\varepsilon$ for which these expressions are valid, given by $\tilde{\varepsilon}= D[\ket{\Psi(s,T_{\rm val})},\ket{\phi_0(s)}]$. Substituting $T_{\rm val}$ from Eq.\eqref{eq:validity_bn} in Eq.\eqref{eq:DpowersTupto2} to first order,
\begin{equation}
\tilde{\varepsilon} = \frac1C \ \frac{\left(\sum_{n\neq0}\left|b_n^{(1)}(s)\right|^2\right)^{3/2}}{\left|\sum_{n\neq0}{\rm Im}\left(b_n^{*(1)}b_n^{(2)}(s)\right)\right|} \ .
\label{eq:epsilontilde}
\end{equation} 

Since we know that error-run-time relations of the type we derived here are valid only for $T \geqslant T_{\rm val}$ and $\varepsilon\leqslant\tilde{\varepsilon}$, we rewrite these relations as
\begin{equation}
\varepsilon\leqslant\alpha\tilde{\varepsilon}\quad\mbox{for}\quad T\geqslant \frac{T_{\rm val}}{\alpha}\,\,\,\,(0<\alpha\leqslant1),
\label{eq:runtimealpha}\end{equation}
and refrain from making any statement for $\varepsilon>\tilde{\varepsilon}$ and/or $T<T_{\rm val}$. Doing so, we guarantee that these relations are correct for any value of $0<\alpha\leqslant 1$.

The values of the coefficients involved can be calculated. For $n\neq0$, $b_n^{(1)}$ can be cast as
\begin{equation}
b_n^{(1)}(s) = \left. e^{iT\omega_{n0}(s')}\lambda_{n0}(s')\right|_{s'=0}^s  \ ,
\label{eq:b1}
\end{equation}
where $\lambda_{nk}(s):=\braket{\phi_n(s)}{\dot\phi_k(s)}/\Delta_{nk}(s)$, in accordance with previous results \cite{MacKenzie2006,Rigolin2008,Wiebe2012}. The coefficient $b_n^{(2)}$, in turn, equals \cite{Rigolin2008}
\begin{multline}
 b_n^{(2)}(s)= e^{iT\omega_{n0}(s)} J_0(s)\lambda_{n0}(s) - J_n(s)\lambda_{n0}(0) + \\
+\left[e^{iT\omega_{n0}(s')}\left(\frac{\dot\lambda_{n0}(s')}{\Delta_{n0}(s')} + \sum_{k\neq0,n}\lambda_{k0}(s')\frac{\braket{\phi_n(s')}{\dot\phi_k(s')}}{\Delta_{n0}(s')}\right) -\sum_{k\neq0,n}e^{iT\omega_{nk}(s')}\lambda_{k0}(0)\lambda_{nk}(s')\right]_{s'=0}^s \ ,
\label{eq:b2}\end{multline}
where $J_n(s):=\sum_{k\neq n}\int_0^s{|\braket{\phi_k(s')}{\dot\phi_n(s')}|^2/\Delta_{kn}(s')}ds'$. 
 The calculation of these coefficients is found in Appendices \ref{sec:leading_coeff} and \ref{sec:nextorderterms}.

\subsection{Boundary cancelation}\label{sec:boundarycancelation}

There are cases of interest, however, where the expressions above are not the most useful, such as when $\dot H(0)=0=\dot H(s)$, which makes $b_n^{(1)}(s)=0$. Especially relevant is the use of boundary cancelation \cite{Garrido1962,Lidar2009,Wiebe2012,Albash2016}, which is known to reduce the scale of the error with $T$ in the asymptotic limit of long times. Therefore, we now tackle Hamiltonian evolutions obeying the boundary-cancelation condition,
\begin{equation}
H^{(j)}(0)=0=H^{(j)}(1) \ \ {\rm for all} \  1\leqslant j\leqslant p \ ,
\label{eq:conds_boundarycancel}
\end{equation}
where $H^{(j)}:=(\partial/\partial s)^j H$. The study of these boundary-cancelation conditions is especially relevant for quantum computation, since the freedom to impose them is typically within the experimenter's reach.
Unlike previous boundary-cancelation results \cite{Garrido1964,Lidar2009,Rezakhani2010a,Wiebe2012}, we go beyond the scaling of the error with $T$ and express the distance in terms of quantities with clear physical interpretations, with concise bounds that only depend on the initial- and final-time Hamiltonian, facilitating numerical approximations. 
 
The expansion of the Bures angle in parameter $1/T$ at the final time $s=1$ is altered, since $b_n^{(1)}(1)=0$. In fact, all coefficients $b_n^{(j)}(1)$ vanish up to $j=p$ and this expansion reads
\begin{equation}
D[\ket{\Psi(1,T)},\ket{\phi_0(1)}] = \frac{\hbar^{p+1}}{T^{p+1}} \sqrt{\sum_{n\neq0}\left|b_n^{(p+1)}(1)\right|^2} 
- \frac{\hbar^{p+2}}{T^{p+2}} 
\frac{\sum_{n\neq0}{\rm Im}\left(b_n^{*(p+1)}(1) \ b_n^{(p+2)}(1)\right)}{\sqrt{\sum_{n\neq0}\left|b_n^{(p+1)}(1)\right|^2}}  + O\left(\frac1{T^{p+3}}\right) \ ,
\label{eq:DpowersTuptop2}
\end{equation}
where $b_0^{(1)}\in\mathbb R$ has been used. This demonstrates the boundary-cancelation effect of producing an error-run-time tradeoff of $\varepsilon\sim1/T^{p+1}$ from the perturbative analysis. Calculating the first non-vanishing coefficient $b_n^{(p+1)}(1)$ and substituting in the expansion, we find
\begin{equation}\begin{split}
\varepsilon = D[\ket{\Psi(1,T)},\ket{\phi_0(1)}] =\frac{\hbar^{p+1}}{T^{p+1}} 
\sqrt{\sum_{n\neq0} \left|\left[ \left.e^{iT\omega_{n0}}\frac{\braket{\phi_n}{H^{(p+1)}|\phi_0}}{\Delta^{p+2}_{n0}}\right|_0^1 \ \right]\right|^2} + O\!\left(\frac1{T^{p+2}}\right) \ . 
\end{split}\label{eq:D_end}\end{equation}
Bounds without the oscillating exponentials can again be found, up to $O(1/{T^{p+2}})$
\begin{align}
\varepsilon =  D[\ket{\Psi(1,T)},\ket{\phi_0(1)}] \lessgtr \frac{\hbar^{p+1}}{T^{p+1}} 
\sqrt{\sum_{n\neq0} \left(\left| \frac{\braket{\phi_n}{H^{(p+1)}|\phi_0}}{\Delta_{n0}^{p+2}}\right|_{1} \pm \left|\frac{\braket{\phi_n}{H^{(p+1)}|\phi_0}}{\Delta_{n0}^{p+2}}\right|_{0}\right)^2} \ .
\label{eq:D_end_bounds}\end{align}
A validity condition is obtained from the comparison of the second term of Eq.\eqref{eq:DpowersTuptop2} with the first, leading to
\begin{equation}
T_{\rm val} = C \ \hbar \left|\frac{\sum_{n\neq0}{\rm Im}\left(b_n^{*(p+1)}(1)\ b_n^{(p+2)}(1)\right)}{\sum_{n\neq0}\left|b_n^{(p+1)}(1)\right|^2}\right| \ .
\label{eq:validity_bn_p}
\end{equation} 
Assuming $p\geqslant1$ and $n\neq0$, $b_n^{(p+1)}$ and $b_n^{(p+2)}$ are given by
\begin{equation}
b_n^{(p+1)}(1) = e^{iT\omega_{n0}(1)}\frac{\lambda_{n0}^{(p)}(1)}{\Delta^{p}_{n0}(1)}-\frac{\lambda_{n0}^{(p)}(0)}{\Delta^{p}_{n0}(0)}
\label{eq:bp1}
\end{equation}
\begin{equation}
b_n^{(p+2)}(1) =  e^{iT\omega_{n0}(1)}\frac{\lambda_{n0}^{(p+1)}(1)}{\Delta_{n0}^{p+1}(1)}
-\frac{\lambda_{n0}^{(p+1)}(0)}{\Delta_{n0}^{p+1}(0)} 
+e^{iT\omega_{n0}(1)}J_0(1)\frac{\lambda_{n0}^{(p)}(1)}{\Delta_{n0}^{p}(1)}
-J_n(1)\frac{\lambda_{n0}^{(p)}(0)}{\Delta_{n0}^{p}(0)}
\label{eq:bp2}
\end{equation}
and we remark that the conditions for boundary cancelation of Eq.\eqref{eq:conds_boundarycancel} allow us to write $\lambda_{n0}^{(j)}=-\braket{\phi_n|H^{(j+1)}}{\phi_0}/\Delta_{n0}$ for $s=0$ or $1$ and $j\leqslant p+1$.

The error value $\tilde{\varepsilon}$ that upper-bounds the validity of these expressions is given by substituting Eq.\eqref{eq:validity_bn_p} in Eq.\eqref{eq:DpowersTuptop2} up to leading order:
\begin{equation}
\tilde{\varepsilon} = \frac1{C^{p+1}} \ \frac{\left(\sum_{n\neq0}\left|b_n^{(p+1)}(1)\right|^2\right)^{p+3/2}}{\left|\sum_{n\neq0}{\rm Im}\left(b_n^{*(p+1)}(1)\ b_n^{(p+2)}(1)\right)\right|^{p+1}} \ .
\label{eq:epsilontildep}
\end{equation} 
Given that $b_n^{(p+1)}$ and $b_n^{(p+2)}$ are of the same order of magnitude --- as will be the case in our example --- this expression shows a general, disadvantageous feature of boundary cancelation: although a very favorable scaling of the error with the run time is eventually obtained, that scaling is only reached for extremely low values of the error ($\propto1/C^{p+1}$ for large $C$). As will be clear in our application to Grover search, practical error tolerances will often be above this bound, and, other drivings present acceptably low errors in a time shorter than $T_{\rm val}$. Notice also that the trade-off is now valid for
\begin{equation}
\varepsilon\leqslant\alpha^{p+1}\tilde{\varepsilon}\quad\mbox{and}\quad T\geqslant \frac{T_{\rm val}}{\alpha}\,\,\,\,(0<\alpha\leqslant1) .
\label{eq:runtimealphap}\end{equation}

Finally, a comment on the quantity $J_n$, which is present in $T_{\rm val}$ [compare Eqs.(\ref{eq:validity_bn_p},\ref{eq:bp2}), or Eqs.(\ref{eq:validity_bn},\ref{eq:b2})]. This quantity is upper-bounded by
\begin{equation}
J_n(s)\leqslant\int_0^s\frac{\mathcal F_{\phi_n}(s')ds'}{4\min_k\Delta_{kn}(s')} \ ,
\label{eq:boundJnFisher}
\end{equation}
where $\mathcal F_{\phi_n}(s)$ is the quantum Fisher information of eigenstate $\ket{\phi_n(s)}$ with respect to $s$. It has been established that $\sqrt{\mathcal F_{\phi_n}(s)}$ constitutes a speed for the evolution of the state $\ket{\phi_n(s)}$ when using the Fubini-Study metric \cite{Braunstein1994,Taddei2013} and therefore is a figure of merit for the driving speed of the Hamiltonian $H(s)$. This relation establishes a connection between the adiabatic theorem and the differential geometry of the quantum state space, which will be explored as we apply these results to the Grover search.

\section{Application: Adiabatic quantum search}\label{sec:grover}
In the following, we apply the general results  presented in the previous section  to a concrete example given by the adiabatic quantum search algorithm,  first presented in Ref.~\cite{Farhi2000} and later modified in Ref.~\cite{Roland2002}. It is an adiabatic version of Grover's original discrete quantum search algorithm~\cite{Grover1997}. Grover's algorithm searches for a marked item in an unstructured database composed of $N$ items. It finds the marked item with error probability $O(1/\sqrt N)$ after $O(\sqrt N)$ queries to a quantum ``oracle'', whereas a classical algorithm has to query an ``oracle'', on average, $N/2$ times in order to find the marked item. This means that the quantum algorithm leads to a quadratic speedup in solving this problem when compared to the best classical algorithm.

We present in this Section error-run-time relations with accurate values whose scaling has never been reported. The correct assessment of the validity time $T_{\rm val}$ is crucial for these results. By finding a driving schedule that minimizes this validity time $T_{\rm val}$, we independently obtain the optimal driving from \cite{Roland2002}. We are the first to report that, under such driving, one obtains $\varepsilon=O(1/\sqrt N)$ after time $T=O(\sqrt N)$, equaling the discrete-Grover scaling. The evidence that $T=O(\sqrt N)$ was sufficient to achieve small error in this case was so far only numerical \cite{Albash2016} and the scaling $\varepsilon=O(1/\sqrt N)$ was unheard of in the adiabatic search. Later, for the simple linear interpolation, we show for the first time that the usual $T=O(N)$ run time leads to a very small $\varepsilon=O(1/N^{3/2})$ error. For the boundary-cancelation scenario, where $p$ derivatives are set to zero, we are the first to obtain error values and to show that, after time $T=O(N)$ with pre\-fac\-tor increasing with $p$, a scaling of $\varepsilon=O(1/N^{p+3/2})$ is reached. This points to an unsuitability of boundary cancelation for AQC. We also numerically comp\-a\-r\-e our result to other bounds in the literature and to numerical simulations, showing that our bounds are tight. In summary, we show that the polynomial regime of $\varepsilon\sim1/T$, $\varepsilon\sim1/T^{p+1}$ is fully capable of giving the correct scaling of error and run time with system size --- once the proper validity time is known. This is in contrast to existing claims in the literature~\cite{Rezakhani2010a}.

\subsection{The adiabatic quantum search algorithm}\label{subsec:adiabaticgrover}
The adiabatic quantum search algorithm is an important example of quantum speedup obtainable  via AQC. It uses $n=\log_2N$ q-bits to represent the $N$ items of the database. Each item is represented by one of the elements $\ket{x}$ of the computational basis, with $x=0,\cdots ,N-1$, while the marked item is represented by $\ket{m}$. The oracle Hamiltonian is given by $ H_f=\mathbbm{1}-\ketbra{m}{m}$, which has $\ket{m}$ as ground state with eigenvalue $0$. The $n$ q-bits are initialized in the uniform superposition state $\ket{\sigma}=\frac{1}{\sqrt{N}}\sum_{x=0}^{N-1}\ket{x}$ and are submitted to the action of a time-dependent Hamiltonian which interpolates between the initial Hamiltonian $H_i=\mathbbm{1}-\ketbra{\sigma}{\sigma}$, which has $\ket{\sigma}$ as ground state with eigenvalue $0$,  and the final Hamiltonian $H_f$:
\begin{align}
H(s)/\mathsf{E}&=\left(1\!-\!f(s)\right)H_i+f(s) H_f \label{eq:searchhamilt}\\ \nonumber
&=\left(1\!-\!f(s)\right)\left(\mathbbm{1}\!-\!\ketbra{\sigma}{\sigma}\right)+f(s)\left(\mathbbm{1}\!-\!\ketbra{m}{m}\right)\ ,
\end{align}
where $\mathsf{E}$ is a dimensional constant that establishes the energy scale and $f(s)$ is a monotonically increasing function with $f(0)=0$ and $f(1)=1$. Since the system is initialized in the state $\ket{\sigma}$, its evolution will remain in the subspace spanned by $\ket{\sigma}$ and $\ket{m}$. The problem may, then,  be studied in the two-dimensional subspace spanned by $\ket{m}$ and $\ket{m_\perp}=\sqrt{\frac{N}{N-1}}\left(\ket{\sigma}-1/\sqrt{N}\ket{m}\right)$. In this subspace, the Hamiltonian $H(s)/\mathsf{E}$ reduces to
\begin{equation}
\!\!\!\!\left.\frac{H(s)}{\mathsf{E}}\right|_{\ket{m},\ket{m_\perp}}\!=\!\left(\begin{array}{cc}\!\!\left(\frac{N\!-\!1}{N}\right)\left(1\!-\!f(s)\right)&\sqrt{\frac{N\!-\!1}{N^2}}\left(f(s)\!-\!1\right)\\
\!\!\sqrt{\frac{N\!-\!1}{N^2}}\left(f(s)\!-\!1\right)& \left(\frac{N\!-\!1}{N}\right)f(s)\!+\!\frac{1}{N}
\end{array}\right).
\end{equation}
Its dimensionless eigenvalues inside this subspace are given by
\begin{equation}
E_0(s)=\frac{1}{2}\left(1-\Delta(s)\right),\qquad  E_1(s)=\frac{1}{2}\left(1+\Delta(s)\right),
\label{eq:eigenvaluesgrover}\end{equation}
where
\begin{equation}
\Delta(s)=\sqrt{1-4\left(\frac{N-1}{N}\right)f(s)\left(1-f(s)\right)}
\label{eq:gapgrover}\end{equation}
is the dimensionless gap between the two eigenvalues. The corresponding eigenstates can be written as
\begin{eqnarray}
\ket{\phi_0(s)}=\sin{\frac{\theta(s)}{2}}\ket{m}+\cos{\frac{\theta(s)}{2}}\ket{m_\perp} \ , \ \ \ \ \  
\ket{\phi_1(s)}=\cos{\frac{\theta(s)}{2}}\ket{m}-\sin{\frac{\theta(s)}{2}}\ket{m_\perp},
\label{eq:eigenstatesgrover}\end{eqnarray}
where
\begin{eqnarray}
\sin{\frac{\theta(s)}{2}}=\frac{1}{\sqrt{2}}\left(1-\frac{2\left(\frac{N-1}{N}\right)\left(1-f(s)\right)-1}{\Delta(s)}\right)^{1/2} \ , \ \ \ \ \ 
\cos{\frac{\theta(s)}{2}}=\frac{1}{\sqrt{2}}\left(1+\frac{2\left(\frac{N-1}{N}\right)\left(1-f(s)\right)-1}{\Delta(s)}\right)^{1/2}\!\!\!\!\!.
\label{eq:sincos}\end{eqnarray}
The remaining $N-2$ eigenstates of $H(s)/\mathsf{E}$ outside the two-dimensional subspace have a common  constant eigenvalue $E_2=1$. 

Later in this section we shall need the transition matrix element $\bra{\phi_1(s)}\dot{H}(s)/\mathsf{E}\ket{\phi_0(s)}$, which can be  determined straightforwardly  from Eqs.~\eqref{eq:searchhamilt} and \eqref{eq:eigenstatesgrover}-\eqref{eq:sincos}:
\begin{equation}
\bra{\phi_1(s)}\dot{H}(s)/\mathsf{E}\ket{\phi_0(s)}=-\frac{\sqrt{N-1}}{N}\frac{\dot{f}(s)}{\Delta(s)}.
\label{eq:matrixhdot}
\end{equation}

\subsection{Error-run-time relations}\label{subset:erorrtimerel}
We apply now the general results of Section~\ref{sec:general} to the adiabatic quantum search algorithm in order to obtain a trade-off relation between the error $\varepsilon$ and the run time $T$ for this problem. First we shall consider the traditional interpolation schedules, which do not  require that  derivatives of $H(s)$ vanish at $s=0$ and $s=1$. We begin by adapting the expression for the upper bound to the error $\varepsilon$ in Eq.~\eqref{eq:D_along_bounds} and the expression for $T_{\rm val}$ in Eq.~\eqref{eq:validity_bn} to the two-dimensional problem. They become
\begin{widetext}\begin{eqnarray}
&\varepsilon&\leqslant \frac{\hbar}{\mathsf{E} T}\left(\frac{\left|\bra{\phi_1(1)}\dot{H}(1)/\mathsf{E}\ket{\phi_0(1)}\right|}{\Delta^2(1)}+
\frac{\left|\bra{\phi_1(0)}\dot{H}(0)/\mathsf{E}\ket{\phi_0(0)}\right|}{\Delta^2(0)}\right) + O\!\left(\frac1{T^2}\right)\label{eq:errsearch}\\
&T_{\rm val} &= C  \frac{\hbar}{\mathsf{E}} \left|\frac{{\rm Im}\left(b_1^{*(1)}(1)\ b_1^{(2)}(1)\right)}{\left| b_1^{(1)}(1)\right|^2}\right|
=C  \frac{\hbar}{\mathsf{E}} \left|{\rm Im}\left(\frac{b_1^{(2)}(1)}{b_1^{(1)}(1)}\right)\right|.\label{eq:tvalsearch}
\end{eqnarray} \end{widetext}
The coefficients $b_1^{(1)}(1)$ and $b_1^{(2)}(1)$ can be obtained from Eqs.~\eqref{eq:b1} and \eqref{eq:b2}, respectively. They read
\begin{eqnarray}
b_1^{(1)}(1)&=&  e^{iT\omega_{10}(1)}\lambda_{10}(1)-\lambda_{10}(0)\label{eq:b1search}\\
b_1^{(2)}(1)&=&e^{iT\omega_{10}(1)}\left[J_0(1)\lambda_{10}(1)+\dot{\lambda}_{10}(1)\right]
								-\left[J_1(1)\lambda_{10}(0)+\dot{\lambda}_{10}(0)\right],
\label{eq:b2search}
\end{eqnarray}
where
\begin{eqnarray}
\lambda_{10}(s)&=&\frac{\braket{\phi_1(s)}{\dot{\phi}_0(s)}}{\Delta(s)}=-\frac{\bra{\phi_1(s)}\dot{H}(s)/\mathsf{E}\ket{\phi_0(s)}}{\Delta^2(s)}=\frac{\sqrt{N-1}}{N}\frac{\dot{f}(s)}{\Delta^3(s)}\label{eq:lambdasearch}\\
J_0(1)&=&-J_1(1)=\int_0^1\frac{\left|\braket{\phi_1(s)}{\dot{\phi}_0(s)}\right|^2}{\Delta(s)}\, ds\,,
\label{eq:j1search}\end{eqnarray}
and we have used that $\Delta(0)=\Delta(1)=1$. In the two interpolation examples which we shall consider in the following, the properties of $f(s)$ are such that $\lambda_{10}(0)=\lambda_{10}(1)$ and $\dot{\lambda}_{10}(0)=-\dot{\lambda}_{10}(1)$. For these most common choices of $f(s)$, we have
\begin{eqnarray}
b_1^{(1)}(1)&\!=&  \lambda_{10}(1)\left(e^{iT\omega_{10}(1)}-1\right)\label{eq:b1search2}\\
b_1^{(2)}(1)&\!=&\!\left(J_0(1)\lambda_{10}(1)+\dot{\lambda}_{10}(1)\right)\!\!\left(e^{iT\omega_{10}(1)}+1\right)\!.
\label{eq:b2search2}
\end{eqnarray}
Notice that for $T=T_n=2n\pi/\omega_{10}(1)$, where $n=1,2,\cdots$, the coefficient $b_1^{(1)}(1)$ vanishes. This means that at those specific run times
the leading term of $D[\ket{\Psi(1,T)},\ket{\phi_0(1)}]$ in Eq.~\eqref{eq:DpowersTupto2}, which is of the order of $1/T_n$, vanishes and the leading term becomes of the order of $1/T_n^2$. This happens because of a boundary symmetry of the Hamiltonian $H(s)$,  which is manifest in the condition $\lambda_{10}(1)=\lambda_{10}(0)$. The use of such classes of symmetries  as a technique to improve the scaling of the error $\varepsilon$ with $T$, at certain discrete values of $T$, was proposed in Ref.~\cite{Wiebe2012} and applied to the adiabatic search problem. While the authors of Ref.~\cite{Wiebe2012} could only confirm the improvement of the scaling of $\varepsilon$ with $T$ from numerical evidence, we can easily calculate the error $\varepsilon(T_n)$ at the times $T_n$ via Eqs.~\eqref{eq:DpowersTupto2},\eqref{eq:D_along},\eqref{eq:b1search2} and \eqref{eq:b2search2}:
\begin{eqnarray}
\varepsilon(T_n)=&-&\frac{\hbar^2}{\left(\mathsf{E} T_n\right)^2}\frac{{\rm Im}\left(b_1^{*(1)}(1)b_1^{(2)}(1)\right)}{\left| b_1^{(1)}(1)\right|}+ O\!\left(\frac1{T_n^3}\right)\\
=&2&\frac{\hbar^2}{\left(\mathsf{E} T_n\right)^2}\left(J_0(1)\lambda_{10}(1)\!+\!\dot{\lambda}_{10}(1)\right)+ O\!\left(\frac1{T_n^3}\right).
\end{eqnarray}
Since the change in the scaling of $\varepsilon$ will happen only at discrete values of $T$, we shall ignore this oscillatory behavior and use an upper bound to the value of $\varepsilon$ given by Eq.~\eqref{eq:errsearch}, which can be rewritten as
\begin{equation}
\varepsilon\leqslant \frac{\hbar}{\mathsf{E} T}2\left|\lambda_{10}(1)\right|+O\!\left(\frac1{T^2}\right),
\label{eq:errsearch2}\end{equation}
where we made use of Eq.~\eqref{eq:lambdasearch} and of the relation $\lambda_{10}(1)=\lambda_{10}(0)$.
By the same token, we will disregard another zero-measure set of times $T$, that of when $b_n^{(2)}(1)=0$.
In order to correctly handle the oscillatory behaviors of $b_1^{(1)}(1)$ and $b_1^{(2)}(1)$  in the evaluation of $T_{\rm val}$, we shall take the contributions of the oscillatory terms to $b_1^{(1)}(1)$ and $b_1^{(2)}(1)$ in Eqs.~\eqref{eq:b1search2} and \eqref{eq:b2search2} at their point of maximum absolute value.
 Under these conditions, we reach
 \begin{equation}
 T_{\rm val} =  \frac{\hbar}{\mathsf{E}}\, C\! \left| J_0(1)+\frac{\dot{\lambda}_{10}(1)}{\lambda_{10}(1)}\right|.
 \label{eq:tvalsearch2}\end{equation}
 
We are now in a position to derive explicit trade-off relations between the error $\varepsilon$ and the run time $T$. Let us first consider the second term in Eq.~\eqref{eq:tvalsearch2}. This term depends only on the boundary conditions of the problem, that is, it depends only on the value of the parameters for $s=0$ and $s=1$. Making use of Eqs.~\eqref{eq:gapgrover} and \eqref{eq:lambdasearch}, we can write 
\begin{equation}
\frac{\dot{\lambda}_{10}(1)}{\lambda_{10}(1)}=\frac{\ddot{f}(1)}{\dot{f}(1)}-6\left(\frac{N-1}{N}\right)\dot{f}(1).
\label{eq:lambdotlamb}\end{equation}

 Let us now consider more closely the term $J_0(1)$. Using Eq.~\eqref{eq:j1search}, $J_0(1)$ can be rewritten as
\begin{equation}
J_0(1)=\frac{1}{4}\int_0^1\frac{\mathcal{F}_{\phi_0}(s)}{\Delta(s)} ds\,,
\label{eq:j1fisher}\end{equation}
where
\begin{eqnarray}
\mathcal{F}_{\phi_0}(s)&=&4\left|\braket{\phi_1(s)}{\dot{\phi}_0(s)}\right|^2
=4\frac{\left|\bra{\phi_1(s)}\dot{H}(s)/\mathsf{E}\ket{\phi_0(s)}\right|^2}{\Delta^2(s)}
=\left[2\frac{\sqrt{N-1}}{N}\frac{\dot{f}(s)}{\Delta^2(s)}\right]^2
\label{eq:fishersearch}\end{eqnarray}
is the  quantum Fisher information of the instantaneous ground state $\ket{\phi_0(s)}$ with respect to $s$. 
As stated previously,  $\sqrt{\mathcal{F}_{\phi_0}(s)}$ is proportional to the instantaneous speed of separation between two neighboring states $\ket{\phi_0(s)}$ and $\ket{\phi_0(s+ds)}$  when using the Fubini-Study metric \cite{Braunstein1994,Taddei2013}. That is,  $\sqrt{\mathcal{F}_{\phi_0}(s)}$ is proportional to the instantaneous speed of the ground state $\ket{\phi_0(s)}$ when moving along the path joining the states $\ket{\phi_0(0)}$ and $\ket{\phi_0(1)}$ in the state space. This shows that $J_0(1)$ is a geometrical quantity, in a differential geometric sense. It is natural, then, to investigate the properties of the path, determined by $f(s)$, which minimizes $J_0(1)$ and, consequently, $T_{\rm val}$.  For this purpose, consider the quantity
\begin{equation}
K=\frac{1}{4}\int_0^1 \mathcal{F}_{\phi_0}(s)\, ds\, ,
\label{eq:actionsearch}\end{equation}
which is called the action of the path followed by $\ket{\phi_0(s)}$ when moving from $\ket{\phi_0(0)}$ to $\ket{\phi_0(1)}$. Using the Cauchy-Schwarz inequality, it is straightforward to show that
\begin{equation}
K\geqslant \left(\int_0^1 \sqrt{\mathcal{F}_{\phi_0}(s)}/2 \, ds\right)^2=\mathcal{L}^2,
\end{equation}
where $\mathcal{L}$ is the Bures length, as defined by Uhlmann~\cite{Uhlmann1992}, of the  path in the state space followed by $\ket{\phi_0(s)}$. Equality in the above expression occurs when the integrand is constant, that is, when $\mathcal{F}_{\phi_0}(s)$ is constant. We shall minimize $J_0(1)$ by minimizing the action $K$. This happens when
\begin{equation}
\sqrt{\frac{\mathcal{F}_{\phi_0}(s)}{4}}=\mathcal{L}_{\rm sh}\!=\! D[\ket{\phi_0(0)},\ket{\phi_0(1)}] = \arccos{\left(1/\sqrt{N}\right)},
\label{eq:fisherminimunaction}\end{equation}
where $\mathcal{L}_{\rm sh}$ is the Bures length of the shortest path between $\ket{\phi_0(0)}$ and $\ket{\phi_0(1)}$. By definition, it is equal to the distance $D$ [Eq.~\eqref{eq:def_Bures}] between $\ket{\phi_0(0)}$ and $\ket{\phi_0(1)}$ --- i.e. between $\ket{\sigma}$ and $\ket{m}$ ---, which is in fact attained by this interpolation.
  Notice that this choice of $\mathcal{F}_{\phi_0}(s)$ forces the state $\ket{\phi_0(s)}$ to move with constant speed along a geodesic in the state space. From Eq.~\eqref{eq:fishersearch}, one can see that, in order to maintain $\mathcal{F}_{\phi_0}(s)$ constant  along the whole evolution, the rate of change of the Hamiltonian $H(s)$ has to be adapted to the local value of the gap $\Delta(s)$. When the gap decreases, the Hamiltonian $H(s)$ has to change slower, whereas, when the gap increases,  $H(s)$ changes faster. 

We look now for the interpolation $f(s)$ which satisfies condition~\eqref{eq:fisherminimunaction}. Using Eq.~\eqref{eq:fishersearch}, one can show that, in order to satisfy condition~\eqref{eq:fisherminimunaction}, $f(s)$ has to be the solution of 
\begin{equation}
\dot{f}(s)=\frac{N}{\sqrt{N-1}}\arccos{\left(1/\sqrt{N}\right)}\Delta^2(s),
\label{eq:condminact}\end{equation}
with the conditions $f(0)=0$ and $f(1)=1$ [notice that $\Delta$ depends on $s$ through $f(s)$, see Eq.\eqref{eq:gapgrover}]. The solution is given by
\begin{equation}
f(s)=\frac{1}{2}\left[1+\frac{1}{\sqrt{N-1}}\tan{\left(\arctan{\left(\sqrt{N-1}\right)} (2s-1)\right)}\right],
\label{eq:fminaction}\end{equation}
where we used $\arccos{(1/\sqrt{N})}=\arctan{\sqrt{N-1}}$. This interpolation was first proposed in Ref.~\cite{Roland2002} as a solution to a ``local adiabatic condition'' and was found to be optimal. More recently,  in Refs.~\cite{Rezakhani2009a,Rezakhani2010,Rezakhani2010a}, the authors used differential geometric techniques to obtain the same interpolation as an adiabatic geodesic on a Hamiltonian-parameter space, unlike our state-space geometric approach.

We now determine the value of $T_{\rm val}$ when the interpolation in $H(s)$ is given by Eq.~\eqref{eq:fminaction}. We use the auxiliary relations 
\begin{subequations}
\begin{align}
\frac{\dot{\lambda}_{10}(1)}{\lambda_{10}(1)}=&-2\arccos{\left(1/\sqrt{N}\right)}\sqrt{N-1} 
\label{eq:lambdotlamb2} \\
J_0(1)=&\arccos{\left(1/\sqrt{N}\right)}\sqrt{N-1} \ , \label{eq:J0arccos}
\end{align}
the former obtained from Eqs.~\eqref{eq:lambdotlamb} and \eqref{eq:condminact}, the latter using Eq.~\eqref{eq:condminact} for a change of variables in the integration. \end{subequations} Substituting these in Eq.~\eqref{eq:tvalsearch2}, we get finally 
\begin{equation}
 T_{\rm val} =  \frac{\hbar}{\mathsf{E}}\, C\arccos{\left(1/\sqrt{N}\right)}\sqrt{N-1}.
 \label{eq:tvalsearchminact}\end{equation}

The upper bound to the error $\varepsilon$ [Eq.~\eqref{eq:errsearch2}], for the optimal interpolation~\eqref{eq:fminaction}, can be straightforwardly calculated with the use of Eq.~\eqref{eq:lambdasearch}:
\begin{equation}
\varepsilon\leqslant \frac{\hbar}{\mathsf{E} T}\,2\arccos{\left(1/\sqrt{N}\right)}+O\!\left(\frac1{T^2}\right).
\label{eq:errsearchminact}\end{equation}
Substituting $T_{\rm val}$ in the above expression, we obtain $\tilde{\varepsilon}$ [Eq.~\eqref{eq:epsilontilde}], which is the upper bound to $\varepsilon$ at $T=T_{\rm val}$ and establishes the largest value of $\varepsilon$ for which the error-run-time relations in Eqs.~\eqref{eq:errsearch}, \eqref{eq:errsearch2} and \eqref{eq:errsearchminact} are still valid:
\begin{equation}
\tilde{\varepsilon}=\frac{2 }{C\sqrt{N-1}}.
\label{eq:etildesearchminact}\end{equation}

With the important results of Eqs.~\eqref{eq:tvalsearchminact},\eqref{eq:etildesearchminact} at hand, it is important to stress the meanings of $T_{\rm val}$ and $\tilde{\varepsilon}$. Since  bounds on the error $\varepsilon$ like those in Eqs.~\eqref{eq:D_along_bounds}, \eqref{eq:errsearch}, and \eqref{eq:errsearch2} result from a truncation in first order of an expansion of the error $\varepsilon$ in powers of $1/T$, they are not valid for arbitrary values of $T$ and/or $\varepsilon$. Instead of relying on the vagueness of conditions like $T\gg1$ and/or $\varepsilon\ll1$, we are able to specify, via the value of $T_{\rm val}$, the values of $T$ for which those relations are valid, namely $T\geqslant T_{\rm val}$, and, consequently, $\varepsilon\leqslant\tilde{\varepsilon}$.

Notice that, for $N\gg1$, $\tilde{\varepsilon}\approx\frac{2}{C\sqrt{N}}=O\left(1/\sqrt{N}\right)$ and $T_{\rm val} \approx \frac{\hbar}{\mathsf{E}} C\frac\pi2\sqrt{N}=O\left(\sqrt{N}\right)$. This means that, in order to reach an error of $\varepsilon=O(1/\sqrt{N})$ with the use of the optimal interpolation, the adiabatic search algorithm will need a run time $T=O(\sqrt{N})$. This is precisely the trade-off relation between error and run time of the original Grover's algorithm~\cite{Boyer1998,Roland2003} in the limit of $N\gg1$. Until now, there had been no reports of $\varepsilon=O(1/\sqrt{N})$ in the adiabatic search, only that with $T=O(\sqrt{N})$ one can reach an error constant in $N$ \cite{Roland2002,Rezakhani2010a}. In fact, evidence that, under the optimal interpolation~\eqref{eq:fminaction}, a run time of $O(\sqrt{N})$ leads to a asymptotically small error $\varepsilon$ has been only numerical~\cite{Albash2016}. Here, we show analytically that this indeed happens, besides presenting a tight bound on the value of $\varepsilon$. 
 
We can contrast these results with previous findings in the literature. The general bounds from \cite{Ambainis2004,Jansen2007,OHara2008} are very rigorous and always valid, but because they do not intend to be tight, they overestimate the run time necessary for a given error (or the error made at a given run time). For instance, Jansen \emph{et al}'s bound \cite[Theorem 3]{Jansen2007}, the tightest of the three, yields asymptotically $\varepsilon\leqslant (\pi/2+\pi^2)\sqrt N/T$. This bound guarantees that in run time $T=O(\sqrt N)$ an error $\varepsilon$ constant in $N$ or lower is attained. This might be sufficiently good for some applications, but we know from Eq.\ (\ref{eq:etildesearchminact}) that the error is much lower for such time, scaling as $\varepsilon=O(1/\sqrt N)$ just like in Grover's algorithm. Conversely, Jansen \emph{et al}'s bound only ensures an error $\varepsilon=O(1/\sqrt N)$ after an overestimated run time of $T=O(N)$.

\begin{figure}%
\includegraphics[width=.55\columnwidth]{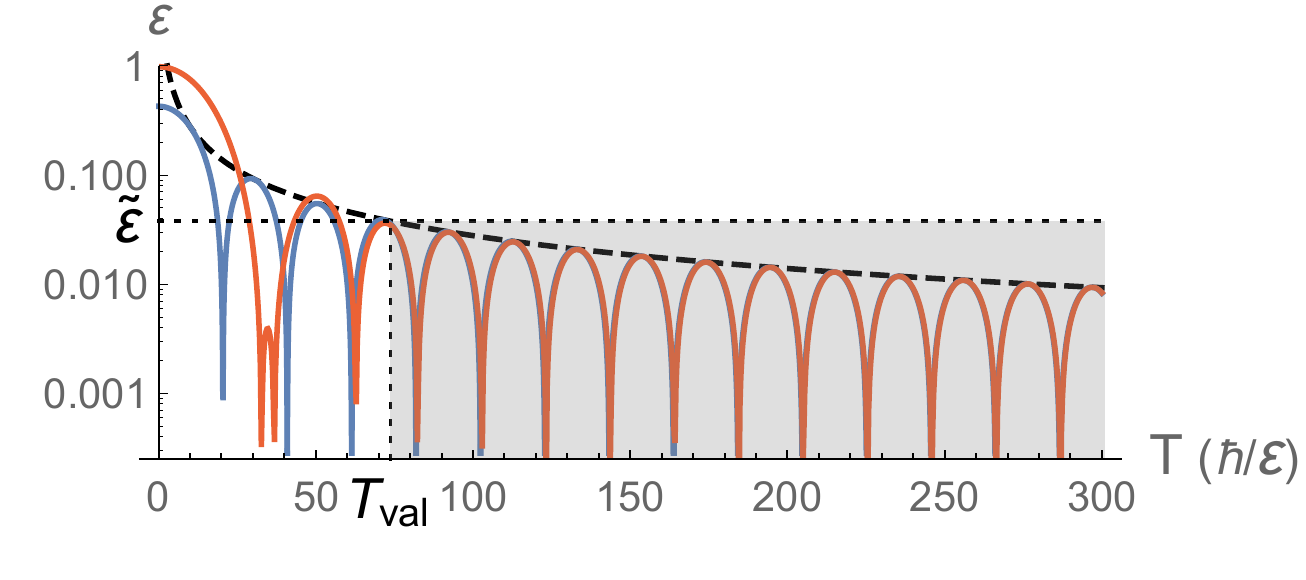}%
\caption{Relation between error and total run time for the adiabatic search using optimal interpolation \eqref{eq:fminaction}, which keeps the speed of the eigenstate (as measured by its quantum Fisher information) constant,  with $N=32$. The leading term, calculated directly from \eqref{eq:D_along} (solid, light gray) is tightly upper-bounded by Eq.\eqref{eq:errsearchminact} (black, dashed). Along with a direct numerical calculation of the distance (solid, black), other results in the literature are shown (dot-dashed): Jansen \emph{et al}'s \cite[Theorem 3]{Jansen2007} (red, above), Roland \emph{et al}'s \cite[Eq.(20)]{Roland2002} (green, below). Rezakhani \emph{et al}'s bound for the short-time regime (blue, dashed) is also to be seen \cite[Eq.(159)]{Rezakhani2010a}. Bounds by Ambainis \emph{et al} \cite[Lemma 3.3]{Ambainis2004} and O'Hara \emph{et al} \cite[Eq.(5)]{OHara2008} are not sufficiently close to the actual values to fit this plot. The upper bound $\tilde\varepsilon$ comes from Eq.\eqref{eq:etildesearchminact}, and it suffices to set $C=9.5$ in Eq.\eqref{eq:tvalsearchminact} for convergence.
}%
\label{fig:constantFisher}%
\end{figure}

 Furthermore, the use of relations between $\varepsilon$ and $T$, like Eq.~\eqref{eq:errsearchminact}, without a clear statement about their definition and domain of validity can be problematic. Roland \emph{et al}~\cite{Roland2002, Roland2003} have taken the optimal interpolation~\eqref{eq:fminaction} and showed that for large $N$ it takes a time $T=(\pi/2)\sqrt N/\varepsilon$ to run such algorithm, where $\varepsilon$ is a parameter attached to the local adiabatic condition with no well-defined relation to the final error. If one tries to naively use this relation as an error-run-time trade-off, one arrives at incorrect results. Although the authors in~\cite{Roland2002, Roland2003} are content to demand a constant $\varepsilon\ll1$, for which no glaring inaccuracy appears, using their relation to seek an error $\varepsilon=O(1/\sqrt{N})$ yields the same overestimation as in the preceding paragraph. However, this is particularly troublesome here, since the authors~\cite{Roland2002,Roland2003} interpret these results as a calculation of the error, not an upper bound as in \cite{Ambainis2004,Jansen2007,OHara2008}.

Rezakhani \emph{et al}'s treatment of the adiabatic search \cite{Rezakhani2010a} also encounters incorrect results, which are readily recognized as such by the authors --- this time, due to lack of validity conditions. When deriving bounds specifically for the adiabatic search problem, the authors of \cite{Rezakhani2010a} discern two regimes: a polynomial one for long run times $T$ and an exponential regime for shorter run times. The former applied to the optimal interpolation coincides with our upper bound Eq.\ \eqref{eq:errsearchminact}, but here is where the validity regime plays a crucial role: by using $\varepsilon$ as a free parameter with no constraints, \cite{Rezakhani2010a} obtains incorrect run times. The situation is equivalent to taking our bound \eqref{eq:errsearchminact} and applying it to an (invalid) error $\varepsilon$ constant in $N$: this gives a run time asymptotically constant in $N$, which is clearly incorrect. The authors of \cite{Rezakhani2010a} recognize the incorrectness but, without validity conditions at their disposal, make a general claim that the polynomial regime is unreliable to obtain a trade-off. We demonstrate the contrary, since once inside the validity regime, proper relations are obtained, such as $T=O(\sqrt N)$ for $\varepsilon=O(1/\sqrt N)$. Rezakhani \emph{et al} resort to the exponential regime as a source of trade-off relations between $\varepsilon$ and $T$. This is a limited approach for a twofold reason. Firstly, even confined to a single application of AQC (the adiabatic search),  this regime does not exist for all interpolations $f(s)$ \cite{Rezakhani2010a}. Secondly, the exponential bound is never tight (Fig.\ref{fig:constantFisher}), leading to overestimations of the run time such as $T=\sqrt N\log N$ for $\varepsilon=O(1/\sqrt N)$.

The expansions in orders of $1/T$ found in \cite{MacKenzie2006,Cheung2011,Wiebe2012} have many merits, arriving at relations equivalent to Eqs.~(\ref{eq:b1},\ref{eq:bp1}), but are also subject to the lack of validity conditions. Attempting to extend their asymptotic results to finite times $T$ to derive trade-off relations falls into the same trap mentioned in the previous paragraph: without clear validity conditions, incorrect results would appear. Cheung \emph{et al} \cite{Cheung2011} additionally obtain bounds of general validity; these, however, are only tight asymptotically (without a clear condition for tightness), and suffer from the same limitations as the previous general bounds, overestimating run times and/or errors.

We have gathered many of the results above for the distance between the instantaneous ground state $\ket{\phi_0(T)}$ and the physical state $\ket{\Psi(1,T)}$, and compared them to our expressions as well as to the distance given by a numerical solution to the Schr\"odinger equation under an optimally driven Hamiltonian [Eq.~\eqref{eq:fminaction}]. This comparison is shown in Fig.\ref{fig:constantFisher}, where one sees that our error value, calculated directly from Eq.~\eqref{eq:D_along}, is very accurate, that our upper bound Eq.~\eqref{eq:errsearchminact} is tight when valid, whereas previous upper bounds never are.

\begin{figure*}%
\includegraphics[width=.95\columnwidth]{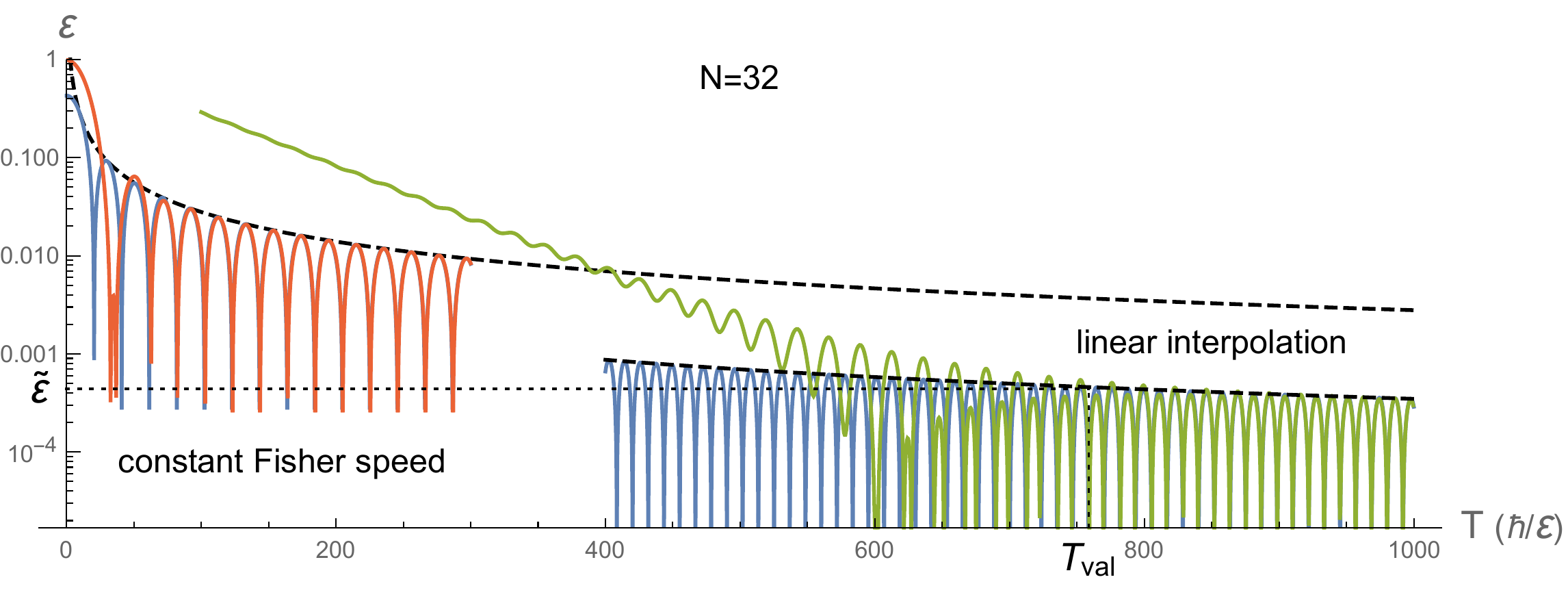}%
\caption{Comparison of error values of the adiabatic search for the linear ($f(s)=s$) and constant-Fisher-speed interpolations. For the latter, parameters are the same as in Fig.\ref{fig:constantFisher} and the oscillation pattern repeats regularly for higher values of $T$ (omitted for $T>300$ for cleanness, where only the bound is shown). For the former, the bound from Eq.\eqref{eq:errsearchlin} (black, dashed) limits the oscillations of the leading-term expression [blue, obtained directly from Eq.\eqref{eq:D_along}], to which the numerical calculation of the distance (green) converges for long enough run time $T$ [$T\geqslant T_{\rm val}$ with $C=50$ in Eq.\eqref{eq:tvalsearchlin}]. Notice that with the linear interpolation the error converges to a smaller value, but requires a longer run time $T$ to do so; for shorter run times, the optimal interpolation \eqref{eq:fminaction} produces lower errors.}%
\label{fig:Fisher_and_linear}%
\end{figure*}

We turn our attention now to the linear interpolation presented in the first proposal of the adiabatic search algorithm~\cite{Farhi2000}. Although it has been shown that this interpolation leads to a run time of $O(N)$ and, for this reason, did not bring any advantage when compared to a classical search  algorithm, it remains the unsolved issue of the value of the error $\varepsilon$ reached at this run time.

For this interpolation, $f(s)=s$. We begin with the  determination of  $T_{\rm val}$. We use the auxiliary relations
\begin{subequations}\begin{equation}
\frac{\dot{\lambda}(1)}{\lambda(1)}=-6\left(\frac{N-1}{N}\right) \ ,
\label{eq:lambdotlamblin}\end{equation}
obtained from Eq.~\eqref{eq:lambdotlamb}, $\dot{f}(s)=1$, and $\ddot{f}(s)=0$; and 
\begin{equation}
J_0(1)=\!\left(\frac{N-1}{N}\right)\!\int_0^1\!\frac{ds}{\Delta^5(s)}=\frac{1}{3}\left(\frac{N-1}{N}\right)+\frac{2}{3}\left(N-1\right) \ ,
\end{equation}\end{subequations}
due to Eqs.~\eqref{eq:j1fisher} and \eqref{eq:fishersearch}. Putting these two expressions into Eq.~\eqref{eq:tvalsearch2} produces finally 
\begin{equation}
 T_{\rm val} =  \frac{\hbar}{\mathsf{E}}\, C\frac{1}{3}\left|2\left(N-1\right)-17\left(\frac{N-1}{N}\right)\right|.
 \label{eq:tvalsearchlin}\end{equation}
The upper bound to the error $\varepsilon$ [Eq.~\eqref{eq:errsearch2}]  can be obtained with the use of Eq.~\eqref{eq:lambdasearch}:
\begin{equation}
\varepsilon\leqslant \frac{\hbar}{\mathsf{E} T}\frac{2\sqrt{N-1}}{N}+O\!\left(\frac1{T^2}\right).
\label{eq:errsearchlin}\end{equation}
Substituting $T_{\rm val}$ in the above expression, we obtain $\tilde{\varepsilon}$ [Eq.~\eqref{eq:epsilontilde}], which is the upper bound to $\varepsilon$ at $T=T_{\rm val}$:
\begin{equation}
\tilde{\varepsilon}=\frac{6 }{C\left|(2N-17)\sqrt{N-1}\right|}.
\label{eq:etildesearchlin}\end{equation}
 
Now we can analyze the trade-off relation from Eqs.~\eqref{eq:tvalsearchlin},\eqref{eq:etildesearchlin}. For $N\gg1$, $\tilde{\varepsilon}\approx3/(CN^{3/2})=O(1/N^{3/2})$ and $T_{\rm val}\approx\frac{\hbar}{\mathsf{E}} C2N/3=O(N)$. This means that, to reach an error $\varepsilon=O(1/N^{3/2})$, the use of the linear interpolation requires a run time $T=O(N)$. This is the well known linear scaling of the run time with $N$ for the linear interpolation. It has been derived already in Refs.~\cite{Farhi2000,Roland2002} as $T=O(N/\varepsilon)$ and in Ref.~\cite{Rezakhani2010a} as $T=O(N\ln{(1/\varepsilon}))$. Notice that, if one requires an error $\varepsilon=O(1/N^{3/2})$, both these relations predict a run time very different from $O(N)$, showing, once more, the problem of having the error $\varepsilon$ as a free parameter. Fig.\ref{fig:Fisher_and_linear} compares our analytical results with the corresponding quantities produced by a numerical solution of the Schr\"odinger equation, confirming the validity of our results.

Let us compare the linear and optimal interpolations. For fixed value of $N$, the linear interpolation has a larger value of $T_{\rm val}$ than the optimal one. 
However, we can ask about the run time required by the optimal interpolation~\eqref{eq:fminaction} to reach the error value that the linear interpolation reaches with a run time $T=O(N)$. This can be easily answered with the use of Eq.~\eqref{eq:runtimealpha}, with $\alpha$ set as the ratio of the values of $\tilde{\varepsilon}$ for both interpolations.
 For $N\gg1$, the optimal interpolation reaches such error after time
 \begin{equation}
 T=\frac{T_{\rm val}}{\alpha}=\frac{2N}{3}T_{\rm val}= \frac{\hbar}{\mathsf{E}}\, C\frac{\pi N^{3/2}}{3} \ ,
 \label{eq:ratiolinopt}\end{equation}
where $T_{\rm val}$ comes from Eq.~\eqref{eq:tvalsearchminact} (optimal interpolation). Interestingly, this time is $\sqrt{N}$ larger than the time required by the linear interpolation (cf. Fig.\ref{fig:Fisher_and_linear}). Interpolation~\eqref{eq:fminaction} is optimal because of its optimal scaling with $N$ for any asymptotically small error, but this result shows that different goals (such as a particularly low error scaling) require different strategies.


We consider now Hamiltonian evolutions obeying the boundary-cancelation condition~\eqref{eq:conds_boundarycancel}. A Hamiltonian $H(s)$ which obeys this condition can be built with a function $f(s)$ obeying $f^{(j)}(0)=0=f^{(j)}(1)$ for $j=1,...,p$. The minimal-order polynomial that satisfies these conditions up to order $p$ can be written as \cite{Rezakhani2010a,Wiebe2012}
\begin{equation}
f_p(s):=\frac{\int_0^s x^p(1-x)^pdx}{\int_0^1 x^p(1-x)^pdx}=I_s(p+1,p+1) \ ,
\label{eq:fp}
\end{equation}
where $I_s(p+1,p+1)$ is the regularized incomplete $\beta$ function. Notice that $f_p(s)$ satisfies the following relations:
\begin{equation}
f_p^{(p+1)}\!(0)\!=\!(-1)^pf_p^{(p+1)}(1);\ \ \ \ \ f_p^{(p+2)}(0)\!=\!(-1)^{p+1}f_p^{(p+2)}(1).
\end{equation}
 These relations imply that
\begin{equation}
 \lambda_{10}^{(p)}(1)=(-1)^p\lambda_{10}^{(p)}(0)=\frac{\sqrt{N-1}}{N} f_p^{(p+1)}(1).
 \end{equation}
 Now, using the results of Sec.~\ref{sec:boundarycancelation}  and proceeding in the same way as in the previous cases, it is straightforward to show that
  \begin{equation}
 T_{\rm val} =  \frac{\hbar}{\mathsf{E}}\, C\! \left| J_0(1)+\frac{f_p^{(p+2)}(1)}{f_p^{(p+1)}(1)}\right|=
 \frac{\hbar}{\mathsf{E}}\, C\! \left| J_0(1)\!+\!p(p+1)\right|,
 \label{eq:tvalboundary1}\end{equation}
 with
 \begin{align}
 J_0(1)&=\frac{N-1}{N^2}\int_0^1\frac{\dot{f}_p^2(s)}{\Delta^5(s)}\,ds 
 \approx \frac N2 \left(1+\sqrt p +\frac p{20}\right) \ , \label{eq:j1approx}
\end{align}
where the $O(N)$ dependence is valid in general for $N\gg1$, and the approximate value including the dependence on $p$ is correct up to 5\% error. One can then write:
 \begin{equation}
 T_{\rm val} \approx  \frac{\hbar}{\mathsf{E}}\, C\! \left| \frac N2 \left(1+\sqrt p +\frac p{20}\right)\!+\!p(p+1)\right|.
 \label{eq:tvalboundary2}\end{equation}
 Notice that $T_{\rm val}$ for this case has the same scaling with $N$ as $T_{\rm val}$ for the linear interpolation case.
 
\begin{figure*}[htb]%
\includegraphics[width=.95\columnwidth]{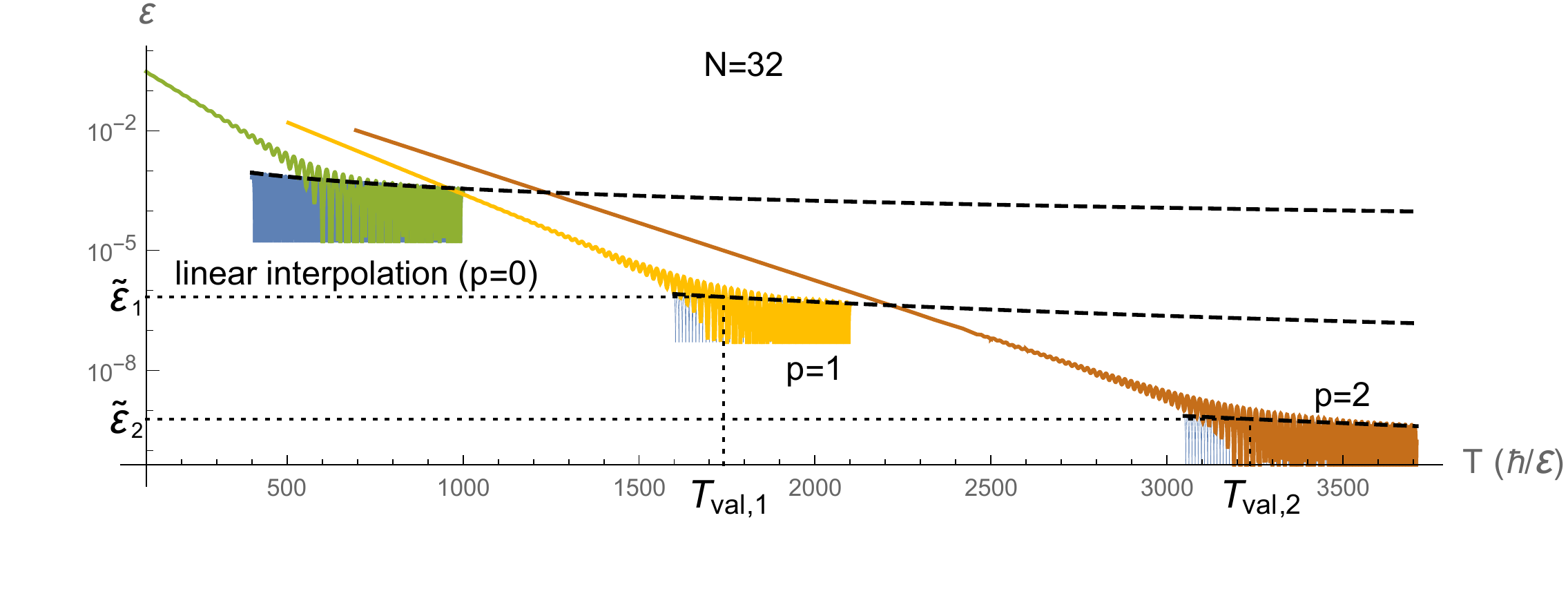}%
\caption{Behavior of the error in adiabatic search using boundary cancelation. Upper bounds [from Eq.\eqref{eq:errboundary}], distance up to leading order [from Eq.\eqref{eq:D_along}] and direct numerical calculation are displayed as in Fig.~\ref{fig:Fisher_and_linear}. With each successive order $p$ a smaller error value is attained, as seen in Eq.\eqref{eq:epsilontildeboundary}, at the cost of a longer $T_{\rm val}$. These validity times are obtained from Eq.\eqref{eq:tvalboundary2} by setting $C=50$ for $p=1$ and $C=70$ for $p=2$.}%
\label{fig:boundarycancelation}%
\end{figure*}

 The upper bound on the error $\varepsilon$ can be obtained via the same proceeding:
 \begin{align}
\varepsilon\leqslant \frac{\hbar^{p+1}}{(\mathsf{E} T)^{p+1}}\,\frac{2\sqrt{N-1}}{N}\frac{(2p+1)!}{p!} \ ,
\label{eq:errboundary}\end{align}
where we used  $|f_p^{(p+1)}(1)|=(2p+1)!/p!$.
 Using Eqs.~\eqref{eq:tvalboundary2} and \eqref{eq:errboundary}, we can determine $\tilde\varepsilon$, the upper bound to $\varepsilon$ at $T=T_{\rm val}$ (for $N\gg1$):
 \begin{equation}
 \tilde{\varepsilon}\approx\frac1{C^{p+1}}\frac{2^{p+2}(2p+1)!}{(1+\sqrt p+p/20)^{p+1}p!}\frac1{N^{p+3/2}}.
 \label{eq:epsilontildeboundary}\end{equation}
  Therefore, in order to reach an error $\varepsilon=O\left(N^{p+3/2}\right)$, the use of an interpolation which satisfies the  boundary-cancelation condition~\eqref{eq:conds_boundarycancel} requires a run time $T=O(N)$, with a prefactor increasing with $p$. This means that this interpolation reaches a much smaller final error than the linear one, but requires a longer run time, although with the same scaling in $N$.	Fig. \ref{fig:boundarycancelation} compares the behavior of the error $\varepsilon$ with the run time $T$ for the different interpolations we have considered here.
  Notice that this also happens in comparison with the optimal interpolation~\eqref{eq:fminaction}, but to a higher degree:  for run times larger or equal to $T=O(N)$, the error reached by the current interpolation is much smaller than that reached by the optimal interpolation, at the price of a longer necessary convergence time. For application in AQC, where a constant error typically suffices and shortest time is the ultimate goal, boundary cancelation typically presents a disadvantage.

\section{Final comments}\label{sec:conclusions}

We have obtained errors of the adiabatic approximation for finite driving times $T$ beyond scaling relations.  Using the Bures angle between the actual state of the system and the instantaneous ground state of the Hamiltonian at time $T$ as a figure of merit for the error, we derived  general results for the trade-off relation between the run time $T$ and the error $\varepsilon$  in terms of quantities with very clear physical interpretations.  With the use of the Adiabatic Perturbation Theory (APT), we performed an expansion of the error $\varepsilon$ for small values of $1/T$ up to next-to-leading-order terms.  This allowed us to delimit the domain of validity of the trade-off relations obtained via leading-order approximations and, consequently, to explicitly determine the shortest time $T=T_{\rm val}$ and the largest error $\tilde\varepsilon=\varepsilon$ for which these relations are valid.  When the use of these relations is restricted to such values, they correctly give the value of the error $\varepsilon$ for a given run time $T$ or, conversely, the value of the run time $T$ required for reaching a given error $\varepsilon$. In this way, the results go beyond scaling-only relations. This approach corrects some problems present in several trade-off relations derived so far, where the error $\varepsilon$ appears as a ``free'' parameter, although those relations rely on some form of approximation, hence are valid only for a restricted range of $\varepsilon$ (and $T$).
 
 We applied these  results to the adiabatic quantum search algorithm, which, besides being of  practical importance for AQC,  admits an exact analytical treatment. With our approach, we were able to  reobtain the interpolation which is considered optimal for this algorithm  via a minimization of the functional for the shortest time $T_{\rm val}$ for which 
the relations we derived still apply. More importantly, apart from getting the right scaling of the run time with the system size $N$ --- showing for the first time that Grover-like $\varepsilon=O(1/\sqrt N)$ is attained in time $T=O(\sqrt N)$ ---, we could determine the actual value of the final error for that interpolation. Novel values (and scalings) for errors reached at a given run time were also found for other classes of interpolation, such as the linear one [$\varepsilon=O(1/N^{3/2})$ at $T=O(N)$] and the scheme named boundary cancelation [$\varepsilon=O(1/N^{p+3/2})$ at $T=O(N)$]. 
In doing so, we could show that traditional trade-off relations, as much as they can give results sufficient for practical use in famous applications (such as $\varepsilon=O(1)$ in the adiabatic search), are inherently not tight and overestimate the error and/or time for a given evolution. With the expansion of AQC it is only more likely that such overestimation misses important quantum advantages in the future.

Our results are in no way restricted to the adiabatic version of the Grover search. Applications can be made to test the known adiabatic approach to other quantum problems, such as satisfiability problems \cite{Farhi2000,Farhi2001}, the Deutsch-Josza problem \cite{Das2002,Wei2006}, the Bernstein-Vazirani problem \cite{Hen2014}, or Simon's problem \cite{Hen2014}. We believe that, in these and other cases, our approach may offer an improvement for deriving trade-off relations between error and run time for quantum adiabatic algorithms, allowing for precise estimations of the run time required for an algorithm to be performed with a given error.

\begin{acknowledgments}
The authors acknowledge financial support from the Brazilian funding agencies CNPq, CAPES and FAPERJ (grant E-26/202.801/2016), and the National Institute of Science and Technology for Quantum Information.
\end{acknowledgments}

\appendix{

\section{Expansion of the distance in the parameter $1/T$}
\label{sec:expansionD1T}

In this Appendix, we demonstrate the expansion of the Bures angle in the small parameter $1/T$ up to second nonvanishing order [Eqs.(\ref{eq:DpowersTupto2},\ref{eq:DpowersTuptop2})]. In fact, there is a two-fold dependence of the Bures angle on $T$. On one hand, there is the dependence expressed in the expansion of Eq.\eqref{eq:bn_expansion},
\begin{equation}
b_n(s,T) = \sum_{p=0}^\infty \frac{(i\hbar)^p}{T^p} \ b_n^{(p)}(s) \ .
\label{eq:bn_expansion_v}
\end{equation}
On the other hand, $b_n^{(p)}(s)$ may still depend on $T$ (through a complex exponential). We will make a perturbative expansion of $D[\ket{\Psi(s,T)},\ket{\phi_0(s)}]$ in the small parameter $\frac1T$, which should not be confused with a genuine power-series expansion. This means that the expansion is made as if $b_n^{(p)}(s)$ did not depend on $T$. We assume $b_n^{(0)}=\delta_{n0}$ (recovering the infinite-time result), and also that we calculate the distance at a time ${s}$ at which the coefficients $b_n^{(q)}({s})=0$ for $n\neq0$ and $q=1,2,...,p$. Eq. (\ref{eq:DpowersTupto2}) will correspond to $p=0$ (for any value of ${s}$), and for the boundary-cancelation case ($p\geqslant1$), this assumption will be proven for ${s}=1$ in Appendix \ref{sec:leading_coeff}.

Let us express the Bures angle between the instantaneous ground state $\left|\phi_0(s)\rangle \right.$ and the state of the system $|\psi(s,T)\rangle$. Because a truncated expression for $|\psi(s,T)\rangle$ does not preserve normalization \cite{Rigolin2008} (as usual in perturbation theories), we replace Eq.\eqref{eq:def_Bures}, written for normalized states, for an expression not assuming a normalized $|\Psi(s,T)\rangle$ \cite{Kobayashi1969,Taddei2014},
\begin{equation}
D[\ket{\Psi(s,T)},\ket{\phi_0(s)}] := \arccos 
\left(\frac{\left| \braket{\phi_0(s)}{\Psi(s,T)} \right|}{\sqrt{\braket{\Psi(s,T)}{\Psi(s,T)}}}\right) \ ,
\label{eq:def_Bures_nonnorm}
\end{equation}
where $|\phi_0(s)\rangle$ is still assumed normalized. In terms of the coefficients in Eq.\eqref{eq:Psi_expansion},
\begin{equation}
D[\ket{\Psi(s,T)},\ket{\phi_0(s)}] = \arccos 
\left(\frac1{\sqrt{1+\sum_{n\neq0}\frac{|b_n(s,T)|^2}{|b_0(s,T)|^2}}}\right) \ .
\label{eq:Bures_bn}
\end{equation}
We define $w := \sum _{n\neq0} |b_n(s,T)|^2/| b_0(s,T)|^2$, which, together with a trigonometric identity, allows us to write
\begin{equation}
D[|\Psi (s,T)\rangle ,|\phi _0(s)\rangle]  = \arctan \sqrt{w}.    
\label{eq:arctgw}
\end{equation}
We can expand \eqref{eq:arctgw} in a power series as 
\begin{equation}
\label{arctgexpanded}
D[|\Psi(s,T)\rangle,|\phi _0(s)\rangle] = w^{1/2} - \frac{1}{3} w^{3/2} + O(w^{5/2}),   
\end{equation}
which converges absolutely for $|w|< 1$  or $D[|\Psi(s,T)\rangle,|\phi _0(s)\rangle] < \pi/4 $.

From Eq.\eqref{eq:bn_expansion_v}, we find that
\begin{align}
w=&(\hbar/T)^{2p+2}
\sum_{n\neq0}\frac{\left|b_n^{(p+1)}\right|^2+\hbar/T\ 2{\rm Re}(i\ b_n^{*(p+1)}\ b_n^{(p+2)})+O(1/T^{2})}{1+2{\rm Re}(i\hbar/T \ b_0^{(1)}) + O(1/T^2)} \ ,\\
=&\frac{\hbar^{2p+2}}{T^{2p+2}}\sum_{n\neq0}\left(|b_n^{(p+1)}|^2 + \hbar/T \left[2{\rm Re}(i\ b_n^{*(p+1)}\ b_n^{(p+2)}) -\left|b_n^{(p+1)}\right|^2\ 2{\rm Re}(i\hbar/T \ b_0^{(1)})\right]  + O(1/T^2)\right) \ ,
\label{eq:w_powers_v}\end{align}
where $b_n^{(p+1)}$ is the first nonvanishing coefficient for $n\neq0$, but $b_0^{(1)}\neq0$ in general. We see that $w^{1/2}=O(1/T^{p+1})$ and the two leading terms will be of order $1/T^{p+1}$, $1/T^{p+2}$, with $w^{3/2}=O(1/T^{3p+3})$ being neglected (since $1/T^{3p+3}$ is higher-order than $1/T^{p+2}$ for any $p\geqslant0$). Then,
\begin{multline}
D[|\Psi(s,T)\rangle,|\phi _0(s)\rangle] = \frac{\hbar^{p+1}}{T^{p+1}}\sqrt{\sum_{n\neq0}\left|b_n^{(p+1)}\right|^2} \times\\
 \left\{1+2\frac\hbar T \left[\frac{\sum_{n\neq0}{\rm Re}(i\ b_n^{*(p+1)}\ b_n^{(p+2)})}{\sum_{n\neq0}\left|b_n^{(p+1)}\right|^2}  - {\rm Re}\left(i\frac\hbar T \ b_0^{(1)}\right)\right]+O\left(\frac{1}{T^2}\right)\right\}^{1/2} +O\left(\frac1{T^{3p+3}}\right) \ ,
\label{eq:Dpowersv}
\end{multline}
and finally,
\begin{multline}
D[|\Psi(s,T)\rangle,|\phi _0(s)\rangle] = \frac{\hbar^{p+1}}{T^{p+1}}\sqrt{\sum_{n\neq0}\left|b_n^{(p+1)}\right|^2}+\\ 
+ \frac{\hbar^{p+2}}{T^{p+2}} \left[\frac{\sum_{n\neq0}{\rm Re}(i\ b_n^{*(p+1)}\ b_n^{(p+2)})}{\sqrt{\sum_{n\neq0}\left|b_n^{(p+1)}\right|^2}}  - \sqrt{\sum_{n\neq0}\left|b_n^{(p+1)}\right|^2} \ {\rm Re}\left(i\frac{\hbar}T \ b_0^{(1)}\right)\right] +O\left(\frac{1}{T^{p+3}}\right) \ , 
\label{eq:Dpowersv_final}\end{multline}
which yields Eq.\eqref{eq:DpowersTupto2} for $p=0$. For $p\neq0$, Eq.\eqref{eq:DpowersTuptop2} is reached by using that $b_0^{(1)}$ is real (proven in Appendix \ref{sec:leading_coeff}).

\section{Leading coefficients}
\label{sec:leading_coeff}

In this Appendix, we calculate the leading coefficients of expansion \eqref{eq:bn_expansion} both for the general case, $b_n^{(1)}(s)$, and with boundary cancelation, $b_n^{(p+1)}(1)$ --- both for $n\neq0$. Along the way we demonstrate that the assumptions made in the previous Appendix are valid ($b_n^{(0)}=\delta_{n0}$ and $b_n^{(q)}({s})=0$ for $n\neq0$ up to $q=p$).

\subsection{General Case}
We now need to calculate $b_n^{(1)}(s)$ for $n\neq0$. We begin by writing the coefficients in Eq.\eqref{eq:bn_expansion} as
\begin{equation}
b_n^{(p)}(s)=\sum _{m=0} e^{iT{\omega}_{nm}(s)}b_{nm}^{(p)}(s) \ \ .
\label{eq:coeff}
\end{equation}
Inserting Eqs. (\ref{eq:bn_expansion}),(\ref{eq:coeff}) into Eq.\eqref{eq:Psi_expansion} we get
\begin{equation}\label{eq:three}
|\Psi (s,T)\rangle =\sum _{p=0}^\infty \sum_{n,m=0} \frac{(i\hbar)^p}{T^p}e^{-iT\omega_m(s)}b_{nm}^{(p)}(s)\left|\phi _n(s)\rangle\right. .
\end{equation}
At this point we impose that all the contributions higher than the zeroth order vanish at $s=0$, i.e. $b_n(0,T)=b_n^{(0)}(0)$ and 
\begin{equation}
\label{initialone}
b_n^{(p)}(0) = \sum_{m=0}b_{nm}^{(p)}(0)=0, \quad (p \geqslant 1) .    
\end{equation}
Furthermore we also impose that the zeroth order contribution is the adiabatic approximation,
\begin{equation}
\label{initialtwo}
b_n^{(0)}(s) = b_n^{(0)}(0)=\delta_{n0} \Rightarrow b_{nm}^{(0)}(s) = b_{nm}^{(0)}(0) = \delta_{n0}\delta_{nm} \ .
\end{equation}
We now substitute \eqref{eq:three} into the Schrödinger Equation [Eq.\eqref{eq:Schroedinger}] and left-multiply it by $\bra{\phi _k(s)}$, leading to the general recurrence relation 
\begin{equation}
\label{fundamental}
\Delta_{nm}(s)b_{nm}^{(p+1)}(s)=\dot{b}_{nm}^{(p)}(s) +\sum_{k\neq n}M_{nk}(s)b_{km}^{(p)}(s) ,
\end{equation}
where $M_{nk}(s) := \braket{\phi _n(s)}{\dot{\phi }_k(s)}$ and $M_{nn}(s)$ has been taken as zero without loss of generality.

Let us then calculate the terms  $b_{nm}^{(1)}(s)$, starting by the off-diagonal terms ($n\neq m$). From Eq.\eqref{initialtwo} and Eq.\eqref{fundamental} with $p=0$,
\begin{equation}
\label{first}
b_{nm}^{(1)}(s) = \frac{M_{nm}(s)}{\Delta_{nm}(s)}\delta_{m0} = \lambda_{nm}(s)\delta_{m0} \ , \quad   (n\neq m) .
\end{equation} 
The recurrence relation \eqref{fundamental} furnishes a first-order differential equation for the diagonal terms, 
\begin{equation}
\label{wanted}
\dot{b}_{nn}^{(1)}(s)+\sum_{k\neq n}M_{nk}(s)b_{kn}^{(1)}(s) = 0,
\end{equation}
which can be integrated
\begin{equation}
\label{integral}
b_{nn}^{(1)}(s) =  b_{nn}^{(1)}(0) - \int _0^s\underset{k\neq n}{\sum}M_{nk}(s')b_{kn}^{(1)}(s')ds'
\end{equation}
and, for $n\neq0$, solved using Eqs.\eqref{first}, \eqref{initialone}:
\begin{equation}
\label{zero}
b_{nn}^{(1)}(s)=b_{nn}^{(1)}(0) = -\sum_{m\neq n}b_{nm}^{(1)}(0) =  -\lambda_{n0}(0).
\end{equation}
From \eqref{eq:coeff}, \eqref{first} and \eqref{zero}, one obtains the first-order contribution $b_n^{(1)}(s)$ as in Eq.\eqref{eq:b1}, in accordance with \cite{Rigolin2008}. One then substitutes that in Eq.\eqref{eq:DpowersTupto2} and, taking into account that
\begin{equation}
M_{nm}(s) = -\frac{\langle \phi_n(s)|\dot{H}(s)|\phi _m(s)\rangle}{\Delta_{nm}(s)} \ ,
\label{eq:M_as_Hdot}
\end{equation}
Equation \eqref{eq:D_along} is obtained. (Notice that equality \eqref{eq:M_as_Hdot} is demonstrated by differentiating the eigenvalue equation.)

We take the opportunity to calculate $b_{00}^{(1)}(s)$ from \eqref{integral}, which will be necessary for a later result. We first note that, due to \eqref{initialone} and \eqref{first},
\begin{equation}
b_{00}^{(1)}(0) = - \sum_{m\neq0} b_{0m}^{(1)}(0) = 0 \ .
\label{eq:initialb001}
\end{equation}
Substituting \eqref{first} in \eqref{integral} then gives
\begin{equation}
b_{00}^{(1)}(s) =  - \int _0^s\sum_{k\neq 0}M_{0k}(s')\frac{M_{k0}(s')}{\Delta_{k0}(s')}ds'=J_0(s) \ ,
\label{eq:b001}
\end{equation}
where $M_{0k}=-M_{k0}^*$ has been used and 
\begin{equation}
J_n(s)=\sum_{k\neq n}\int_0^s \frac{|M_{kn}(s')|^2}{\Delta_{kn}(s')}ds' \ ,
\label{eq:Jn}
\end{equation}
as defined in the main text.

\subsection{Boundary cancelation}

We will introduce several Lemmas to show that $b_n^{(r)}(1)$ vanish up for $r=1,2,...,p$ and to calculate the value of $b_n^{(p+1)}(1)$.
\begin{lem}
For $n\neq m$ the general form of the coefficients $b_{nm}^{(p)}(s)$ is given by
\begin{equation}
\label{ansatz}
b_{nm}^{(p)} (s) = \sum _{q=0}^{p-1}  \sum_{k\neq m}\chi_{nmk}^{p,q}(s)\lambda _{km}^{(q)}(s) \ .
\end{equation}
where $\lambda _{km}^{(q)}(s) := \frac{d^q\lambda _{km}(s)}{ds^q}$, and $\chi _{nmk}^{p,q}(s)$ are coefficients to be obtained recursively.
\label{th:chi}\end{lem}
\begin{proof}
 The proof is by induction as follows: the statement holds for $p = 1$,
\begin{equation}
\label{quione}
b_{nm}^{(1)}(s) =  \sum _{q=0}^0 \underset{k\neq m
    }{ \sum  }\chi _{nmk}^{1,q}(s)\lambda _{km}^{(q)}(s),
\end{equation}
if we set $\chi_{nmk}^{1,0}(s) = \delta _{kn}\delta _{m0}$, since we reproduce \eqref{first}. Now we prove that if the statement holds for some $p$, then it holds for $(p+1)$. We substitute \eqref{ansatz} into the right-hand side of the recurrence relation \eqref{fundamental} to obtain
\begin{equation}
\label{proofone}
\Delta _{nm}(s)\ b_{nm}^{(p+1)}(s) = 
\sum_{q=0}^{p-1} \sum_{k\neq m} \left( \dot{\chi}_{nmk}^{p,q}(s)\lambda_{km}^{(q)}(s)+ 
 \chi _{nmk}^{p,q}(s)\lambda_{km}^{(q+1)}(s)  
+\sum_{l\neq n,m} M_{nl}(s)\chi_{lmk}^{p,q}(s)\lambda_{km}^{(q)}(s)\right) +M_{nm}(s)b_{mm}^{(p)}(s),
\end{equation}
The last term in \eqref{proofone} can be rewritten as
\begin{equation}
\label{prooftwo}
M_{nm}(s)b_{mm}^{(p)}(s) 
=\sum_{q=0}^{p-1} \sum_{k\neq m} \Delta_{nm}(s) \lambda_{km}^{(q)}(s) b_{mm}^{(p)}(s)\delta_{q0}\delta_{kn},    
\end{equation}
so that this term can be reincorporated into the sum \eqref{proofone} as
\begin{equation}
\label{proofthree}
\Delta _{nm}b_{nm}^{(p+1)} = \sum _{q=0}^{p-1} \sum _{k\neq m} (\dot{\chi}_{nk}^{p,q}\lambda _{km}^{(q)} + \chi _{nk}^{p,q}\lambda _{km}^{(q+1)} + \sum _{l\neq n,m} M_{nl}\chi _{lk}^{p,q}\lambda _{km}^{(q)} + \Delta_{nm}b_{mm}^{(p)}\delta _{q0}\delta _{kn}\lambda _{km}^{(q)}),
 \end{equation}
in which we have omitted the $s$-dependence.
On the other hand, if we make $q\rightarrow q - 1$, then
 \begin{equation}
 \label{prooffour}
 \sum_{q=0}^{p-1} \sum _{k\neq m} \chi _{nmk}^{p,q}(s) \lambda_{km}^{(q+1)}(s) = 
 \sum_{q=0}^{p-1} \sum_{k\neq m} \chi_{nmk}^{p,q-1}(s)\lambda_{km}^{(q)}(s)(1-\delta_{q0})+\!\!\sum_{k\neq m} \chi_{nmk}^{p,p-1}(s)\lambda _{km}^{(p)}(s)
 \end{equation}
 which allows us to rewrite \eqref{proofthree} as
 \begin{equation}
\Delta _{nm}b_{nm}^{(p+1)} = \sum _{k\neq m} \chi _{nmk}^{p,p-1}\lambda _{km}^{(p)} +
\sum _{q=0}^{p-1} \sum _{k\neq m} \left(\dot{\chi }_{nmk}^{p,q}+ \chi_{nmk}^{p,q-1}(1-\delta_{q0}) 
+ \sum _{l\neq n,m} M_{nl}\chi _{lmk}^{p,q} + \Delta_{nm}b_{mm}^{(p)}\delta_{q0}\delta_{kn}\right)\lambda_{km}^{(q)} \ . \label{prooffive}
\end{equation}
Finally after the preceding rearrangements we can compare \eqref{prooffive} with the general form
\begin{equation}
\label{proofsix}
b_{nm}^{(p+1)}=\sum _{q=0}^p \sum _{k\neq m} \chi _{nmk}^{p+1,q}(s)\lambda _{km}^{(q)}(s),
\end{equation}
which fits perfectly if we set 
\begin{equation}
\label{proofeight}
\chi _{nmk}^{p+1,p}(s)=\frac{\chi _{nmk}^{p,p-1}(s)}{\Delta _{nm}(s)},
\end{equation}
and, for $q=0,1,...,p-1$,
\begin{equation}
\label{proofseven}     
\chi _{nmk}^{p+1,q}(s) = \frac{1}{\Delta _{nm}(s)}\left(\dot{\chi }_{nmk}^{p,q}(s) + \chi _{nmk}^{p,q-1}(s)(1-\delta_{q0})
 +\sum _{l\neq n,m} M_{nl}(s)\chi_{lmk}^{p,q}(s) +
\Delta _{nm}(s)b_{mm}^{(p)}(s)\delta_{q0}\delta _{kn}\right),
\end{equation}
which concludes the proof. 
\end{proof}
Lemma \ref{th:chi} leads to a Corollary.
\begin{cor}[Off-diagonal terms]
The condition
\begin{equation}
H^{(1)}(\tilde{s}) = H^{(2)}(\tilde{s}) = \cdots =   H^{(p)}(\tilde{s}) = 0
\label{eq:Hnullderivatives}
\end{equation}
at some time $\tilde s \in [0,1]$ implies, for $n\neq m$,
\begin{equation}
b_{nm}^{(1)}(\tilde{s}) = b_{nm}^{(2)}(\tilde{s}) = \cdots = b_{nm}^{(p)}(\tilde{s}) = 0 \ .
\label{eq:offdiagonalzero}
\end{equation}
\label{th:offdiagzero}
\end{cor}
\begin{proof}
First notice that the Leibniz rule for derivatives provides
\begin{equation}
\label{Leibniz}
\lambda_{nm}^{(q)}(s)=\sum _{j=0}^q
   \frac{q!}{(q-j)!j!}\frac{d^{j}M_{nm}(s)}{{ds}^{j}}
	\frac{d^{q-j}}{ds^{q-j}}\left(\frac{1}{\Delta_{nm}(s)}\right),    
\end{equation}
an expression in which, due to \eqref{eq:M_as_Hdot}, every term involves a derivative of $H(s)$ of order 1 through $(q+1)$. By setting such derivatives of $H(s)$ to zero as in the statement, $\lambda_{nm}^{(q)}(\tilde s)=0$ for $q=0,1,...,p-1$. Due to Lemma \ref{th:chi}, the coefficients $b_{nm}^{(p)}(\tilde s)$ also become zero. 
\end{proof}

At this point we turn to the diagonal terms $b_{nn}^{(p)}(s)$ and to the conditions imposed on the Hamiltonian at the initial time.
\begin{lem}
The condition
\begin{equation}
H^{(1)}(0) = H^{(2)}(0) = ... = H^{(p)}(0) = 0 \ ,
\label{eq:Hnull_initial_derivatives_p}
\end{equation}
implies $\chi_{nmk}^{r,q}(s)\sim \delta_{m0}$ for every coefficient $\chi_{nmk}^{r,q}(s)$ in \eqref{ansatz} with $r\in\{1,2,...,p+1\}$.
\label{th:chideltam0}
\end{lem}
\begin{proof}
The proof is by induction, beginning with the assertion that, for $r=1$, the coefficients are $\chi_{nmk}^{1,q}(s)=\chi_{nmk}^{1,0}(s) = \delta_{kn}\delta_{m0}$, as mentioned after Eq.\eqref{quione}.
The recurrence relations \eqref{proofeight} and \eqref{proofseven} show that if $\chi_{nmk}^{r,q}\sim \delta_{m0}$ ($r\leqslant p$) holds for some value $r$, then it holds for all terms of $\chi_{nmk}^{r+1,q}$ except possibly the one involving $b_{mm}^{(r)}(s)$. But \eqref{fundamental} furnishes a differential equation for the latter which, analogously to Eq.\eqref{integral}, admits the formal solution
\begin{equation}
\label{soleq}
b_{mm}^{(r)}(s) = b_{mm}^{(r)}(0) - \sum_{k\neq m} \int_0^sM_{mk}(s')b_{km}^{(r)}(s')ds'.    
\end{equation}
By hypothesis, $\chi_{nmk}^{r,q}(s')$ are proportional to $\delta_{m0}$ and, due to Lemma \ref{th:chi}, so are $b_{km}^{(r)}(s')$. The initial condition \eqref{initialone} provides
\begin{equation}
\label{initialdifferentialeq}
b_{mm}^{(r)}(0) = - \sum _{k\neq m} b_{mk}^{(r)}(0) = 0 ,    
\end{equation}
where the last equality follows from Corollary \ref{th:offdiagzero} with $\tilde s=0$. As such, $b_{mm}^{(r)}(s)\sim \delta_{m0}$, 
therefore, $\chi_{nmk}^{r+1,q}(s)\sim \delta_{m0}$. The induction breaks down for $\chi_{nmk}^{p+2,q}(s)$ because the initial condition $b_{mm}^{(p+1)}(0)\neq0$.
\end{proof}
We now highlight a result on the diagonal terms $b_{mm}^{(p)}(s)$ that appeared in this demonstration.
\begin{cor}[Diagonal terms]
The condition
\begin{equation}
H^{(1)}(0) = H^{(2)}(0) = ... = H^{(p)}(0) = 0
\label{eq:Hnull_initial_derivatives}
\end{equation}
implies, for $m\neq0$, 
\begin{equation}
b_{mm}^{(1)}(s) = b_{mm}^{(2)}(s) = ... = b_{mm}^{(p)}(s) \equiv  0 \ .
\label{eq:diagonalzero}
\end{equation}
\label{th:diagonalzero}
\end{cor}
Corollaries \ref{th:offdiagzero} and \ref{th:diagonalzero}, when substituted in the expansion \eqref{eq:coeff}, lead to an important conclusion:
\begin{lem}[Vanishing coefficients]
If, for any $\tilde{s}\in[0,1]$ and all $r\in\{1,2,...,p\}$,
\begin{equation}
H^{(r)}(0) = 0 = H^{(r)}(\tilde s) \ ,
\label{eq:Hnull_init_final}
\end{equation}
then, for $n\neq0$,
\begin{equation}
b_n^{(1)}(\tilde s) = b_n^{(2)}(\tilde s) = ... = b_n^{(p)}(\tilde s) = 0 \ .
\label{eq:bnnull}
\end{equation}
\end{lem}

We are then left with the task of calculating the leading contribution in these conditions, namely, the one in lowest power of $1/T$ not identically zero. Since $\lambda_{nm}^{(q)}(\tilde s)=0$ for $q=0,1,...,p-1$ as mentioned in the proof of Corollary \ref{th:offdiagzero}, the expansion \eqref{ansatz} for $b_{nm}^{(p+1)}(\tilde s)$ $(n\neq m)$ reads
\begin{equation}
\label{quasifinalone}
b_{nm}^{(p+1)}(\tilde{s})= \sum _{k\neq m} \chi _{nmk}^{p+1,p}(\tilde{s})\lambda _{km}^{(p)}(\tilde{s}) \ .
\end{equation}
The coefficients $\chi_{nmk}^{p+1,p}(\tilde{s})$ can be obtained by iterating \eqref{proofeight} up to $\chi_{nmk}^{1,0}(s) = \delta_{kn}\delta_{m0}$:
\begin{equation}
\label{quasifinaltwo}
\chi_{nmk}^{p+1,p}(\tilde{s}) = \frac{\delta_{kn}\delta_{m0}}{\Delta_{nm}^p(\tilde{s})} \ .
\end{equation}
For $n = m$ $\neq0$, Eq.\eqref{soleq} and Lemma \ref{th:chideltam0} show that
\begin{equation}
b_{nn}^{(p+1)}(s) = b_{nn}^{(p+1)}(0) \ ,
\label{eq:bnnp1constant}
\end{equation}
and employing \eqref{initialone} together with \eqref{quasifinalone} and \eqref{quasifinaltwo} at $\tilde s=0$, we have 
\begin{equation}
\label{quasifinalthree}
b_{nn}^{(p+1)}(s) = b_{nn}^{(p+1)}(0) = -\sum_{m\neq n}b_{nm}^{(p+1)}(0) = - \frac{\lambda _{n0}^{(p)}(0)}{\Delta^{p}_{n0}(0)}.    
\end{equation}
Eq.\eqref{eq:bp1} for the leading term $b_n^{(p+1)}\left(\tilde{s}=1\right)$ then follows from substituting  \eqref{quasifinalone}, \eqref{quasifinaltwo} and \eqref{quasifinalthree} into \eqref{eq:coeff}. In order to obtain Eq.\eqref{eq:D_end}, one must use
\begin{equation}
\label{lambdap}   
\lambda _{n0}^{(p)}(\tilde s) = - \frac{\langle \phi _n(\tilde s)|H^{(p+1)}(\tilde s)|\phi _0(\tilde s)\rangle }{\Delta _{n0}^{p+2}(\tilde s)},
\end{equation}
at $\tilde s = 0$ and $1$, valid due to the vanishing derivatives of $H(s)$.

\section{Next-to-leading coefficients}
\label{sec:nextorderterms}
Our last proof tackles the next-order coefficients, using the results from App \ref{sec:leading_coeff} above.

\subsection{General Case}

Let us obtain $b_n^{(2)}(s)$. Lemma \ref{th:chi} provides, for $n\neq m$,
\begin{equation}
\label{appbnm2}
b_{nm}^{(2)}(s) = \sum _{k\neq m} \chi_{nmk}^{2,0}(s)\lambda _{km}^{(0)}(s)+\chi_{nmk}^{2,1}(s)\lambda _{km}^{(1)}(s).
\end{equation}
The first order coefficient $\chi _{nmk}^{1,0}(s) = \delta _{nk}\delta_{m0}$ and the recurrence relations  \eqref{proofeight} allow us to obtain the coefficients
\begin{equation}
\label{reccurencetwoonechi}
\chi _{nmk}^{2,1}(s) = \frac{\chi_{nmk}^{1,0}(s)}{\Delta_{nm}(s)} = \frac{\delta_{nk}\delta_{m0}}{\Delta_{nm}(s)} \ ,
\end{equation}
while \eqref{proofseven}, \eqref{zero} and \eqref{eq:b001} lead to
\begin{equation}
\chi _{nmk}^{2,0}(s)= \frac{1}{\Delta_{nm}\left(s\right)}\sum_{l\neq n,m} M_{nl}(s)\delta_{lk}\delta_{m0}+J_0(s)\delta_{kn}\delta_{m0}-\lambda_{m0}(0)\delta_{kn}(1-\delta_{m0})
\label{reccurencetwozerochi}\end{equation}
The nondiagonal term $b_{nm}^{(2)}(s) $ can be written from \eqref{appbnm2}, \eqref{reccurencetwoonechi} and \eqref{reccurencetwozerochi} as
\begin{equation}
\label{appbnm2sol}
b_{nm}^{(2)}(s)= \left(\frac{\dot{\lambda}_{n0}(s)}{\Delta _{n0}(s)} + \sum_{k\neq 0,n} \frac{M_{nk}(s)\lambda_{k0}(s)}{\Delta_{n0}(s)}+J_0(s)\lambda_{n0}(s)\right)\delta_{m0}-\lambda_{m0}(0)\lambda_{nm}(s)(1-\delta_{m0}).    
\end{equation}
We can obtain $b_{nn}^{(2)}(\tilde{s})$ for $n \neq 0$ from \eqref{soleq}, 
\begin{equation}
b_{nn}^{(2)}(s) = b_{nn}^{(2)}(0) - \sum_{k\neq n} \int_0^sM_{nk}(s')b_{kn}^{(2)}(s')ds'\ ,       
\end{equation}
which furnishes
\begin{equation}
\label{appbmmsol}
b_{nn}^{(2)}(s) = -\sum _{m\neq n} b_{nm}^{(2)}(0)- \lambda _{n0}(0)J_n(s),
\end{equation}
where we have made use of \eqref{initialone}. Substituting expressions \eqref{appbnm2sol} and \eqref{appbmmsol} into \eqref{eq:coeff} leads to Eq.\eqref{eq:b2}.

\subsection{Boundary cancelation}

Let us calculate $b_n^{(p+2)}(\tilde s)$ assuming the boundary-cancelation condition of Eq.\eqref{eq:conds_boundarycancel} up to $j=p$ ($p\geqslant 1$). Lemma \ref{th:chi} and \eqref{Leibniz} provide
\begin{equation}
\label{appnondiagonal}
b_{nm}^{(p+2)}(\tilde{s}\text) = 
\sum _{k\neq m}\left(\chi_{nmk}^{p+2,p+1}(\tilde{s})\lambda _{km}^{(p+1)}(\tilde{s})
+\chi_{nmk}^{p+2,p}(\tilde{s}) \lambda_{km}^{(p)}(\tilde{s})\right).    
\end{equation}
Any term of the form $\chi _{nmk}^{r+1,r}(\tilde{s})$ can be obtained by iterating \eqref{proofeight} up to $\chi _{nmk}^{1,0}(\tilde{s})=\delta _{kn}\delta_{m0}$, which gives us 
\begin{equation}
\label{appchione}
\chi_{nmk}^{r+1,r}(\tilde s)=\frac{\delta _{kn}\delta_{m0}}{\Delta_{nm}^{r}(\tilde s)}.  
\end{equation}
The coefficient $\chi _{nmk}^{p+2,p}(\tilde{s})$ can be obtained by iterating \eqref{proofseven}. In this iteration, the factors $\dot\chi_{nmk}^{r+1,r}(\tilde s)$ and $M_{nl}(\tilde s)$ on the right-hand side of \eqref{proofseven} are always null due to our hypothesis of $\dot H(0)=0=\dot H(\tilde s)$. The former can be seen by means of \eqref{appchione} and the Hellmann–Feynman relation ($\dot{E}_n(\tilde s) = \langle\phi _n|\dot{H}(\tilde s)|\phi _n\rangle=0$), the latter from \eqref{eq:M_as_Hdot}.
We are then left with
\begin{equation}
\label{appchitwo}
\chi _{nmk}^{p+2,p}(\tilde{s}) \ =\frac{\chi_{nmk}^{p+1,p-1}(\tilde{s})}{\Delta_{nm}(\tilde{s})} =\cdots =\  \frac{\chi _{nmk}^{2,0}(\tilde{s})}{\Delta_{nm}^p(\tilde{s})}.
\end{equation}
From \eqref{reccurencetwozerochi} with $M_{nl}(\tilde s)=0$, $\lambda_{m0}(0)=0$,
\begin{equation}
\label{appchifour}
\chi_{nmk}^{p+2,p}(\tilde{s})=\frac{J_0(\tilde{s})}{\Delta_{n0}^{p}(\tilde{s})}\delta_{kn}\delta_{m0}.    
\end{equation}
The nondiagonal term $b_{nm}^{(p+2)}(\tilde{s})$ can be written from \eqref{appnondiagonal}, \eqref{appchione} and \eqref{appchifour} as
\begin{equation}
\label{appbnmpplustwo}
 b_{nm}^{(p+2)}(\tilde{s}) = \left(\frac{\lambda_{n0}^{(p+1)}(\tilde{s})}{\Delta_{n0}^{p+1}(\tilde{s})}
+ J_0(\tilde{s})\frac{\lambda_{n0}^{(p)}(\tilde{s})}{\Delta_{n0}^p(\tilde{s})}\right)\delta_{m0}.   
\end{equation}   

Now we obtain the diagonal terms $b_{nn}^{(p+2)}(\tilde{s})$ for $n\neq 0$. We need to solve 
\begin{equation}
\label{appintegral}
b_{nn}^{(p+2)}(s) = b_{nn}^{(p+2)}(0) - \sum_{k\neq n} \int_0^sM_{nk}(s')b_{kn}^{(p+2)}(s')ds'\ . 
\end{equation}
Whereas $b_{nn}^{(p+2)}(0)$ is found using \eqref{initialone} and \eqref{appbnmpplustwo} with $\tilde s=0$,
\begin{equation}
\label{appquasitwo}
b_{nn}^{(p+2)}(0) = -\sum_{k\neq n} b_{nk}^{(p+2)}(0) = - \frac{\lambda_{n0}^{(p+1)}(0)}{\Delta_{n0}^{p+1}(0)},    
\end{equation}
$b_{kn}^{(p+2)}(s)$ needs to be found for any time $s$, not just those times $\tilde s$ when derivatives of the Hamiltonian vanish, as in \eqref{appbnmpplustwo}. Due to \eqref{ansatz}, it depends on $\chi_{knl}^{p+2,q}(s)$, which is obtained from $\chi_{knl}^{p+1,q}(s)$ through \eqref{proofeight} and \eqref{proofseven} (except for $\chi_{knk}^{p+2,0}$). But Lemma \ref{th:chideltam0} guarantees that $\chi_{knl}^{p+1,q}(s)$ is identically zero for $n\neq0$, and the only remaining contribution comes from $\chi_{knk}^{p+2,0}=b_{nn}^{(p+1)}(s)$:
\begin{equation}
\label{essaqueinteressarapa}
b_{kn}^{(p+2)}(s) = b_{nn}^{(p+1)}(s)\lambda_{kn}(s) = \frac{\lambda _{n0}^{(p)}(0)}{\Delta^{p}_{n0}(0)}\lambda_{kn}(s)\ , 
\end{equation}
where \eqref{quasifinalthree} was used. Substituting \eqref{appquasitwo} and \eqref{essaqueinteressarapa} in \eqref{appintegral} and using  \eqref{eq:Jn}, we find, for $n\neq0$,
\begin{equation}
\label{appquasithree}
b_{nn}^{(p+2)}(s) = - \frac{\lambda_{n0}^{(p+1)}(0)}{\Delta_{n0}^{p+1}(0)} -J_n(s)\frac{\lambda_{n0}^{(p)}(0)}{\Delta_{n0}^p(0)}.
\end{equation}
The expression for $b_{n}^{(p+2)}(\tilde{s})$ in Eq.\eqref{eq:bp2} follows from \eqref{eq:coeff}, \eqref{appbnmpplustwo} and \eqref{appquasithree} with $\tilde s=1$, concluding the proof.
}


%

\end{document}